%% file: main.tex
\documentclass[11pt]{article}

\usepackage{style}
\usepackage{shortcuts}

\title{3SUM Is Really Hard: A Real-to-Integer Reduction}
\author{
  Nick Fischer  \\ {\small Max-Planck Institute for Informatics} \and
  Adam Polak    \\ {\small Bocconi University}                   \and
  Jonas Schmidt \\ {\small Bocconi University}
}
\date{}

\begin{document}

\maketitle

\begin{abstract}
\noindent
We show that if the 3SUM problem on integer-valued inputs can be solved in truly subquadratic time, then it can also be solved in truly subquadratic time on real-valued inputs. This answers an open problem posed by Chan, Vassilevska Williams, and Xu [STOC 2022], and constitutes the first such tight real-to-integer self-reduction in fine-grained complexity. Our proof relies on a surprising combination of the Frank--Tardos theorem from optimization with Freiman-type theorems from additive combinatorics.
\end{abstract}

\input{sections/intro}
\input{sections/prelims}
\input{sections/toolkit}
\input{sections/real-to-real-conv}
\input{sections/real-conv-to-int}

\section*{Acknowledgment}

We are very grateful to Karol Węgrzycki for pointing us to the Frank--Tardos theorem~\cite{FrankT87}.

The work of Adam Polak has been funded by the Italian Ministry of University and Research (MUR) under Ministerial Decree No.~18169 of 17 November 2025 -- FIS 3 Call CUP: J53C25002160001

\bibliographystyle{alphaurl}
\bibliography{main}

\end{document}

%% file: sections/intro.tex
\section{Introduction} \label{sec:intro}
Given $n$ numbers, the \emph{3SUM} problem asks whether any three of them sum up to zero. It can be solved in time $O(n^2)$. One of the three main hypotheses of fine-grained complexity \cite{VW18}, the \emph{3SUM Hypothesis}, asserts that this running time is optimal, up to subpolynomial $n^{o(1)}$ factors.

Some mild polylogarithmic improvements over the quadratic running time are however possible. In particular, Baran, Demaine and Pătrașcu~\cite{BaranDP08} showed that relatively simple hashing and bit-packing tricks suffice to obtain an \smash{$O\big(\frac{n^2}{\log^2n}\polyloglog(n)\big)$}-time algorithm for Integer 3SUM, i.e., 3SUM with input numbers required to be integers. It took much more time and a series of papers \cite{GronlundP18,GoldS17,Freund17,Chan20} to arrive at the same running time also for real-valued inputs, and these Real 3SUM algorithms employ a significantly more complex approach based on linear decision trees.

Gajentaan and Overmars~\cite{GajentaanO95} first introduced 3SUM in the context of conditional lower bounds for computational geometry problems, where it is natural to consider real-valued inputs. Later, Pătrașcu~\cite{Patrascu10} showed how to use hashing to reduce Integer 3SUM to combinatorial problems, such as finding triangles in graphs. This technique has quickly become prevalent in 3SUM-based hardness proofs (see~\cite{FischerKP24}), and in modern fine-grained complexity, 3SUM is essentially synonymous with Integer 3SUM.

Chan, Vassilevska Williams, and Xu~\cite{ChanVX22} asked whether fine-grained reductions from Integer 3SUM can be turned into reductions from Real 3SUM, and they made a partial progress on this question. In particular, they give non-tight reductions from Real 3SUM to the Integer All-Numbers 3SUM, All-Color-Pairs Triangle Collection, All-Edges Sparse Triangle, and All-Edges Monochromatic Triangle problems, as well as tight reductions to counting versions of the last two problems. They remark that
\begin{center}\itshape
``(...) it would be extremely interesting if one could show that \\ the Real 3SUM hypothesis implies the Integer 3SUM hypothesis.''
\end{center}
Our main result establishes this implication.

\begin{theorem}[Real-to-Integer Reduction for 3SUM] \label{thm:main}
If Integer 3SUM can be solved in time~$O(n^{2-\epsilon})$, for some constant $\epsilon > 0$, then Real 3SUM can be solved in time $O(n^{2-\Omega(\epsilon^4)})$ (in the comparison-addition model). 
\end{theorem}

In other words, we show that Integer 3SUM is as hard as Real 3SUM. Thus, if we believe that Real 3SUM is really hard, then also Integer 3SUM and all the numerous problems with known (Integer) 3SUM-based fine-grained lower bounds are also really hard. We consider this result quite unexpected. To our knowledge, it constitutes the first and only tight real-to-integer self-reduction in fine-grained algorithms and complexity.

Our result fits a broader line of work that tries to understand how much of the hashing-based machinery developed around 3SUM (and other related problems such as Subset Sum) is actually necessary. One part of this line deals with derandomization~\cite{ChanH20,FischerKP24}, another (including this work) with real-valued inputs~\cite{ChanVX22,Fischer25}.

\paragraph{Further Consequences.} Our reduction from Real to Integer 3SUM goes through an intermediate problem, the Real Convolution 3SUM problem, which asks, for a given input sequence~$a\in\Real^n$, to determine if there exist $i, j, k$ such that $i + j = k$ and $a[i] + a[j] = a[k]$. The integer-valued version of this problem was first introduced by Pătrașcu~\cite{Patrascu10}, who showed that it is equivalent to Integer 3SUM under subquadratic fine-grained reductions, and used it as an intermediate problem for showing hardness of a number of combinatorial problems. Pătrașcu's reduction from 3SUM to Convolution 3SUM relies on linear hashing as in~\cite{BaranDP08}---a technique inherently only suited to integers. It was later simplified and derandomized by Chan and He~\cite{ChanH20}, who ``leave open the question of whether (...) a similar reduction from 3SUM to Convolution 3SUM is possible for real numbers.'' We answer this question in the affirmative with our intermediate reduction (\cref{thm:real-to-real-conv}).

Another important problem in fine-grained complexity is the Exact Triangle problem, which asks to determine whether a given edge-weighted graph contains a triangle whose total edge weight equals exactly zero. In the world of integer-valued inputs, Exact Triangle is known to be harder than both 3SUM~\cite{WilliamsW13} and APSP~\cite{VassilevskaW18} (All-Pairs Shortest Path, the second of the three central problems in fine-grained complexity), and thus its conjectured hardness gives rise to a weaker unifying hypothesis, which turns out to be sufficient to prove hardness of many problems previously known to be 3SUM-hard~\cite{VassilevskaX20b,ChanX24}. Up until now it was not known if there is a real analogue of the reduction from (Integer) 3SUM to (Integer) Exact Triangle. This reduction goes through the Convolution 3SUM problem and applies a generic convolution-to-product transformation~\cite{WilliamsW13}---the latter would work as well with real-valued inputs, but only our \cref{thm:real-to-real-conv} establishes a real counterpart of the former.

\begin{corollary}
If Real Exact Triangle can be solved in time $O(n^{2-\epsilon})$, for some constant $\epsilon > 0$, then Real 3SUM can be solved in time $O(n^{2 - \Omega(\epsilon^2)})$ (in the comparison-addition model).
\end{corollary}

Another immediate consequence of our result is that all problems previously known to be as hard as Integer 3SUM (e.g., All-Pairs Sparse Triangle~\cite{Patrascu10,KopelowitzPP16}) are now also known to be as hard as Real 3SUM.

\subsection{Computational Model} \label{sec:intro:sec:model}
In fine-grained complexity it is customary work in the word RAM model of computation with $O(\log n)$-bit words, which is a natural choice when the input consists of polynomially bounded integers. However, to deal with real-valued inputs, we need to move to a real RAM model with dedicated registers for storing real numbers and performing specific operations on them. Here we need to be particularly careful when describing the set of allowed operations on these real numbers, because reals are inherently infinite-precision and thus the model can easily gain unreasonable computational power. For instance, it is known that when supporting infinite-precision arithmetic operations and flooring allows solving PSPACE-complete problems in polynomial time~\cite{HartmanisS74}. But even without flooring, the real RAM model is unexpectedly powerful. The fine-grained complexity landscape would look very differently in that model; for instance, computing all-pairs shortest paths (in graphs with integer edge weights) is in truly subcubic time~\cite{Yuval76}. While 3SUM with real-valued input may still require quadratic time in such an unreasonably powerful computational model, we prefer to state our reductions in a model as weak as possible.

It turns out that our reduction can be implemented in a very weak model, the \emph{comparison-addition} model. This model, on top of the word RAM model, supports only two real-valued operations: comparisons ($a < b$) and additions---no multiplication, no flooring, no logarithms.

In fact, our reduction works in an even more restricted model which we will refer to as the \emph{realistic} comparison-addition model. This model supports only a single operation, namely to evaluate a comparison of the form
\begin{equation*}
    \sum_{i=1}^m n_i \alpha_i \geq 0,
\end{equation*}
where the $n_i$'s are arbitrary non-zero integers and the $\alpha_i$'s are input reals; each such operation is assumed to take time $(\sum_{i=1}^m \log |n_i|)^{O(1)}$.\footnote{Note that we can still simulate \emph{additions} in the realistic comparison-addition model, by performing all additions \emph{symbolically}. Concretely, define a symbolic number as a tuple $(n_1, \dots, n_m, \alpha_1, \dots, \alpha_m)$ representing $\sum_{i=1}^m n_i \alpha_i$, where the $n_i$'s are integers and the $\alpha_i$'s are pointers to the input reals. Two symbolic numbers can be added in time $(\sum_{i=1}^m \log |n_i|)^{O(1)}$ without invoking any operations on reals. Only a comparison of two symbolic numbers becomes one operation in the realistic comparison-addition model. One can similarly simulate subtractions and multiplications/divisions by integers.} Note that this is exactly the time it would take to evaluate the comparison in the word RAM model if the $\alpha_i$'s were polynomially-bounded integers. In this sense, the realistic comparison-addition model is a safe choice, as it \emph{must} inherit the same fine-grained complexity landscape as the word RAM model for integer inputs. It is possible to impose further restrictions on the model (such as further restricting the number of terms in each sum); we refer to~\cite[Appendix A in the arXiv version]{ChanVX22} for a deeper discussion.

Curiously, the \emph{only} reason why we cannot show that Integer 3SUM is \emph{equivalent} to Real 3SUM in the comparison-addition model is that the simple direction, reducing Integer 3SUM to Real 3SUM, does not work. After all, a hypothetical $O(n^{2-\epsilon})$-time algorithm for Real 3SUM in the comparison-addition model might somehow rely on the unbounded power of the model (e.g., by constructing huge numbers with repeated squaring), and may for this reason not be transferable in the word RAM model. In contrast, the realistic comparison-addition model can be simulated by the word RAM model. For this reason, we do obtain the anticipated equivalence in the realistic comparison-addition model:

\begin{theorem}[Equivalence] \label{thm:equiv}
Integer 3SUM can be solved in the word RAM model in time $O(n^{2-\epsilon})$, for some constant $\epsilon > 0$, if and only if Real 3SUM can be solved in the realistic comparison-addition model in time $O(n^{2-\epsilon'})$, for some constant $\epsilon' > 0$.
\end{theorem}

\subsection{Technical Overview} \label{sec:intro:sec:overview}
In this section, we describe the high-level ideas behind our real-to-integer reduction. It differs fundamentally from the previous non-tight reduction of Chan, Vassilevska Williams, and Xu~\cite{ChanVX22}. Their proof is built around the All-Edges Sparse Triangle problem, which asks to identify all edges that participate in triangles in a given sparse graph. This problem is known to admit a tight reduction from (Integer) 3SUM~\cite{Patrascu10,KopelowitzPP16}, as well as a simple non-tight reduction to 3SUM. One of the main results of~\cite{ChanVX22} is to extend the former reduction to reals, establishing a non-tight reduction from Real 3SUM to All-Edges Sparse Triangle.\footnote{Notably, as a direct corollary of \cref{thm:main}, we obtain a tight reduction from Real 3SUM to All-Edges Sparse Triangle.} By transitivity, this yields a non-tight reduction from Real 3SUM to 3SUM.

This is a conceptually clean approach, but it seems unlikely to yield a tight reduction as in \cref{thm:main}. A major obstacle is that APSP is also known to admit a tight reduction to All-Edges Sparse Triangle~\cite{VassilevskaX20b,ChanVX22}. Consequently, a tight reduction along these lines would also imply a tight connection between 3SUM and APSP---a major open question in fine-grained complexity. For this reason, we depart entirely from this framework, and start over from scratch.

Let $A \subseteq \Real$ denote a Real 3SUM instance. Perhaps the most basic and naive idea is to aim for a simple mapping from $A$ to an equivalent set of integers. Concretely, does there exist a function $h : A \to \Int$ such that
\begin{equation} \label{eq:dream-h}
    a + b + c = 0 \;\;\text{if and only if}\;\; h(a) + h(b) + h(c) = 0?
\end{equation}
And if so, can it be computed efficiently? With access to such a function $h$, the reduction would be trivial: simply replace $A$ by the equivalent Integer 3SUM instance $\set{h(a) : a \in A}$. To our initial surprise, the short answer is: yes, essentially! The detailed answer involves four interesting ideas.

\subsubsection*{Idea 1: Frank--Tardos Theorem}
The first idea, and our starting point, is to exploit the celebrated Frank--Tardos theorem~\cite{FrankT87} from optimization. In its original form, it states that any 0--1 optimization problem can be solved in polynomial time if and only it can be solved in \emph{strongly}\footnote{Informally, an algorithm with integer inputs is said to run in \emph{strongly} polynomial time if it runs in polynomial time, and performs only arithmetic operations (${+}, {-}, {\cdot}, {/}$) and comparisons $({=}, {\leq}, {\geq}, {<}, {>})$ on the integer inputs (and no other operations such as flooring). In particular, all intermediate results are integers that can be stored in polynomial space.} polynomial time. The key technical claim is as follows, intuitively stating that any set of reals can be replaced by (large) integers so that the order of its subset sums is preserved.

\begin{restatable}[Frank--Tardos~\cite{FrankT87}]{theorem}{thmfranktardos} \label{thm:frank-tardos}
There is an algorithm that, given reals $\alpha \in \Real^d$, runs in polynomial time $d^{O(1)}$ and computes integers $v \in \Int^d$ such that
\begin{enumerate}[label=(\roman*)]
    \item $\norm{v}_\infty \leq 2^{O(d^3)}$, and
    \item $\alpha^T x \geq 0$ if and only if $v^T x \geq 0$, for all $x \in \set{-1, 0, 1}^d$.
\end{enumerate}
\end{restatable}

Although our setting differs from that considered by Frank and Tardos, it turns out that this theorem is similarly valid in the comparison-addition model for computing with reals. In \cref{sec:toolkit:sec:frank-tardos} we retrace the proof of \cref{thm:frank-tardos}, following the beautiful exposition by Eisenbrand~\cite{Eisenbrand10}, to verify that all required operations on reals can be implemented as claimed.

Returning to our earlier question, \cref{thm:frank-tardos} directly entails the existence of the desired function $h$. Indeed, let $\alpha_1, \dots, \alpha_n$ denote the elements of $A$, and apply the theorem with $d = n$ to obtain integers $v_1, \dots, v_n$. Then take the function $h$ defined by $h(\alpha_i) = v_i$; \cref{thm:frank-tardos} implies condition~\eqref{eq:dream-h} (as it implies both the $\geq$ and $\leq$ variants of this equivalence).

This is an encouraging starting point, but a fundamental obstacle remains: the Frank--Tardos algorithm is far too slow. It builds on top of the LLL lattice reduction algorithm~\cite{LenstraLL1982}, which itself already runs in time $O(d^6)$, and introduces some more polynomial factors along the way. For $d = n$, this yields a running time of $O(n^6)$, which is well beyond the subquadratic budget $n^{2-\Omega(1)}$ required for our reduction. Even worse, it is \emph{unconditionally impossible} to improve the running time of \cref{thm:frank-tardos} to subquadratic, since the output size is necessarily~$\Omega(d^2)$.\footnote{To see this, let $\alpha$ be the length-$2d$ vector where the first $d$ entries are defined by $\alpha_i = 2^i$ and the last $d$ entries consist of arbitrary integers $\alpha_{d+1}, \dots, a_{2d}$ in the range $[0, 2^d-1]$. Let $v$ denote the integer vector implied by \cref{thm:frank-tardos}. Then one can exactly reconstruct $\alpha_{d+1}, \dots, a_{2d}$ from $v$. Indeed, the bit-representation of $\alpha_i$ is the unique vector $b \in \set{0, 1}^d$ such that $\sum_{j=1}^d v_j b_j = v_i$. It takes $\Omega(d^2)$ bits to specify $\alpha_{d+1}, \dots, a_{2d}$, and thus also $\Omega(d^2)$ bits to specify $v$.}

Overcoming this issue is our key challenge in this paper. Our proposed solution is new and conceptually exciting, but technically quite involved. It relies on several more ideas to be described next.

\subsubsection*{Idea 2: Subset Sum Covers}
Since applying the Frank--Tardos theorem to all $n$ input reals is too expensive, a natural strategy is to invoke it on a much smaller set, say for $d = O(n^{0.1})$. In that case, the polynomial running time becomes subquadratic, which is within our time budget.

To understand which restricted set of reals to use, observe that we have not fully exhausted the power of the Frank--Tardos theorem before: it not only preserves the order of 3-sums $\alpha_i + \alpha_j + \alpha_k$, but in fact of arbitrary \emph{subset sums} $\alpha^T x$ (for $x \in \set{0, 1}^d$). Thus, to effectively use \cref{thm:frank-tardos} it would be sufficient to have access to a vector $\alpha$ so that each element in $A$ can be expressed as a subset sum of $\alpha$; in this case we call $\alpha$ a \emph{subset sum cover} of $A$. We refer to the length of $\alpha$ as the \emph{rank} of the subset sum cover. In this language, our goal is to find a rank-$O(n^{0.1})$ subset sum cover of $A$.

Unfortunately, such a low-rank subset sum cover might not exist---in fact, it typically does not exist! Any set~$A$ without additive structure (e.g., a set of $n$ random reals in some prescribed range) does not admit subset sum covers of rank smaller than $|A|$. Experience shows that these unstructured sets are likely among the hardest instances of 3SUM.

Surprisingly, we still manage to make this approach work. Even if $A$ might not admit a low-rank subset sum cover, we show that an appropriate approximation of $A$ always does. Specifically, we show the following two statements:
\begin{enumerate}[label=(\alph*)]
    \item any sufficiently \emph{additively structured} set has an $O(n^{0.1})$-rank subset sum cover, and \label{itm:goal-struct-implies-cover}
    \item any Real 3SUM instance $A$ can be \emph{forced} to become additively structured by appropriately \emph{rounding} it. \label{itm:goal-force-struct}
\end{enumerate}
In what follows we describe our solutions to these two subgoals.

\subsubsection*{Idea 3: Additive Combinatorics and Freiman's Theorem}
We begin with subgoal~\ref{itm:goal-struct-implies-cover}. Suppose for now that $A$ is additively structured; we show that in this case one can efficiently construct a low-rank subset sum cover of~$A$. To formalize this, we must specify what ``additively structured'' means. Additive combinatorics offers several such notions, and we later make use of multiple variants. For the purposes of this overview, however, we focus on the most basic one: small \emph{doubling}. Let
\begin{equation*}
    A + B \coloneqq \set{a + b : a \in A, b \in B}
\end{equation*}
denote the \emph{sumset} of two sets $A$ and $B$. The \emph{doubling} of a set~$A$ is defined as
\begin{equation*}
    K \coloneqq \frac{|A + A|}{|A|}.
\end{equation*}
This parameter ranges from $1$ to $|A|$. Highly structured sets (e.g., arithmetic progressions $A = \set{a, 2a, \dots, na}$) have $K = O(1)$, whereas unstructured sets (such as random sets) typically have $K$ close to $|A|$.

Our goal is thus to show that any low-doubling set admits a low-rank subset sum cover. This precisely concerns a central result in additive combinatorics, namely Freiman's theorem. To state it, we introduce one more definition. A \emph{generalized arithmetic progression} (GAP) of \emph{rank $d$} is a set of the form
\begin{equation*}
    \set*{\alpha^T x : 0 \leq x_i \leq N_i},
\end{equation*}
for some $\alpha \in \Real^d$ and $N \in \Int_{>0}^d$.
This should be viewed as a generalized set of subset sums of $\alpha$, where each item $\alpha_i$ may be used up to $N_i$ times. Freiman's theorem~\cite{Freiman66} states that every low-doubling set is contained in a low-rank GAP. The following quantitative bounds are due to Ruzsa~\cite{Ruzsa94} and Chang~\cite{Chang02}.\footnote{Historically, Freiman's original proof~\cite{Freiman66} gave weaker quantitative bounds; see also~\cite{Bilu99}. Ruzsa later provided an alternative proof~\cite{Ruzsa94}, which was sharpened by Chang~\cite{Chang02} to obtain the bounds in \cref{thm:freiman}. Some of these results are phrased for integers, but actually work for arbitrary torsion-free abelian groups, and were further generalized to arbitrary abelian groups by Ruzsa and Green~\cite{GreenR07}. Modern formulations often relax containment to covering by a few translates, yielding significantly improved bounds; see, e.g., the milestones~\cite{Sanders12,GowersGMT25}.}

\begin{theorem}[Freiman's Theorem] \label{thm:freiman}
Let $A \subseteq \Real$ be a set with doubling $|A + A| \leq K |A|$. Then $A$ is contained in a generalized arithmetic progression of size $|A| \cdot \exp(K^{O(1)})$ and rank $K^{O(1)}$.
\end{theorem}

Freiman's theorem immediately implies that any set with doubling $|A + A| \leq K|A|$ admits a subset sum cover of rank $K^{O(1)} \log |A|$. Indeed, \cref{thm:freiman} embeds $A$ into a GAP of rank~\makebox{$d = K^{O(1)}$}. Then replace each item~$\alpha_i$ by $\alpha_i, 2\alpha_i, \dots, 2^{\ceil{\log N_i}} \alpha_i$, so that each multiple $x_i \cdot \alpha_i$ for $0 \leq x_i \leq N_i$ can be now be expressed as a subset sum. Consequently, the desired subset sum cover exists.

But again, this argument is purely existential---we still require an efficient algorithm to construct such a cover. While algorithmic versions of Freiman's theorem have been studied for integer sets~\cite{RandolphW24}, extending them to real inputs introduces a substantial obstacle. A key step in proofs over the reals (and more generally torsion-free groups) is a \emph{modeling lemma}, which replaces $A \subseteq \Real$ by an integer set $A' \subseteq \Int$ with equivalent additive structure (a so-called Freiman-isomorphic set); see~\cite[Lemma~5.25]{TaoV06}. This lemma admits a simple non-constructive proof via the structure theorem for finitely-generated torsion-free groups~\cite[Corollary~3.6]{TaoV06}. But here we encounter a circularity: an algorithmic version of the modeling lemma essentially requires the Frank--Tardos theorem. But invoking Frank--Tardos is precisely what we sought to avoid, putting us back to square one.

We follow a different approach. Instead, we devise a new, completely different, and much simpler proof of a weaker version of Freiman's theorem:

\begin{theorem}[Poor Man's Freiman's Theorem] \label{thm:poor-man-freiman-exact}
Let $A \subseteq \Real$ be a set with doubling $|A + A| \leq K |A|$. Then $A$ has a subset sum cover of rank $(K \log |A|)^{O(1)}$.
\end{theorem}

Really the only difference between \cref{thm:freiman,thm:poor-man-freiman-exact} is that in \cref{thm:poor-man-freiman-exact} the rank is allowed to depend polynomially on $\log |A|$ (we omitted the size bound, but it could be guaranteed analogously). For us this dependence on $\log |A|$ is completely acceptable, as we typically think of $K = O(|A|^{0.1})$, say. In contrast, typical additive combinatorics results try to avoid any overhead in terms of $|A|$ (thinking of $K$ as a doubling \emph{constant}). Our more coarse-grained perspective turns out to be a significant advantage here.

In a nutshell, our proof of \cref{thm:poor-man-freiman-exact} is along the following lines. Any small-doubling set~$A$ must have many popular differences, i.e., numbers $x$ that can be expressed in many ways as $a_1 - a_2 = x$ for $a_1, a_2 \in A$. Sample one such difference $x$ uniformly at random, and place $-x$ in the subset sum cover $\alpha$. For each pair $a_1, a_2 \in A$ that can be expressed as $a_1 - a_2 = x$, we can now \emph{discard} $a_2$; if we eventually arrive at a subset sum cover for $a_1$, then we have also covered $a_2$ via $a_2 = a_1 - x$. To cover $a_1$ we simply keep inserting popular differences into $\alpha$.

To analyze this process we think of an initial empty graph over the vertices $A$. With each sampled popular difference we add some random edges $(a_1, a_2)$ to the graph, symbolizing that once we cover one of the two endpoints, we automatically also cover the other endpoint. By an analysis akin to Borůvka's algorithm for Minimum Spanning Tree~\cite{Boruvka28}, we argue that after few samples, namely $(K \log |A|)^{O(1)}$, the graph becomes (almost) fully connected. Once that happens we can take an arbitrary root vertex $a_0 \in A$, add $a_0$ to the subset sum cover, and thereby also cover all other vertices in the graph.

\subsubsection*{Idea 4: Forcing Additive Structure by Rounding}
So far we have described how to handle the case that the given Real 3SUM instance $A$ is additively structured. We now turn to subgoal~\ref{itm:goal-force-struct} and show that we can always transform $A$ into an equivalent, but additively structured, set. At first glance, this may seem paradoxical. Indeed, several prior works show that 3SUM on additively structured instances can be solved in truly subquadratic time~\cite{ChanL15,AbboudBF23,JinX23,Fischer24,FischerJX25}. Thus, if we could simply enforce additive structure without further cost, we would refute the 3SUM Hypothesis outright.

The crux is that the transformation necessarily introduces many \emph{false positives}, that is, triples $a,b,c \in A$ that are not solutions in the original instance but become solutions after the transformation. Consequently, it is not enough to solve the transformed instance---we must also efficiently enumerate and filter a large number of candidate solutions. This phenomenon is standard in 3SUM-based reductions. For example, a simple \emph{universe reduction} maps any integer 3SUM instance to one over a bounded universe $[u]$, but incurs $O(n^3/u)$ false positives. Such reductions typically rely on integer-specific techniques, such as reducing all numbers modulo a random prime in $[u,2u]$.

The natural approach in the real setting is to \emph{round}. To illustrate this idea, consider the following idealized scenario. Suppose there exists some $\Delta \in \Realnneg$ satisfying two properties:
\begin{enumerate}[label=(\roman*)]
    \item the set $\hat A$, obtained from $A$ by rounding each element to its nearest integer multiple of $\Delta$, is additively structured (e.g., $|\hat A + \hat A| \leq K |\hat A|$ for $K = n^{0.1}$), and \label{itm:round-struct}
    \item there are few false positives, $\abs{\set{(a, b, c) \in A^3 : a + b + c \neq 0, a + b + c = 0 \pm \Delta}} \leq n^{1.9}$. \label{itm:round-false-positives}
\end{enumerate}
This situation arises, for example, when $A$ consists of $n$ random reals in $[-1,1]$. Indeed, for $\Delta \approx n^{-1.1}$, the number of false positives is $O(n^{1.9})$. At the same time, $\hat A + \hat A$ has size at most $O(n^{1.1})$ (since there are only $O(n^{1.1})$ multiples of $\Delta$ in $[0,1]$), and hence $\hat A$ has doubling $O(n^{0.1})$.

In this idealized setting, we can replace $A$ with $\hat A$. All true solutions are preserved (up to an additive offset of $\Delta$ or $2\Delta$). By~\ref{itm:round-struct}, the set $\hat A$ is additively structured, and we can therefore compute a subset sum cover using our ``Poor Man’s'' Freiman's theorem. Applying the Frank--Tardos theorem to this cover yields the desired mapping $h : \hat A \to \Int$. We thus obtain an integer instance $\set{h(a) : a \in \hat A}$ that preserves all solutions, at the cost of introducing at most $n^{1.9}$ false positives by~\ref{itm:round-false-positives}. Finally, these can be handled using the standard reduction from listing to decision: if Integer 3SUM can be solved in time $O(n^{2-\epsilon})$, then $n^{1.9}$ solutions can be listed in time $O(n^{2-\epsilon'})$~\cite{VassilevskaW18} (see \cref{lem:3sum-listing-to-decision}). In this idealized regime, the real-to-integer reduction is complete. 

Unfortunately, the idealized setting overlooks two big problems.

\paragraph{Problem 1: Rounding}
The first obstacle is that we cannot explicitly construct $\hat A$ in the comparison-addition model---or, in fact, in any ``reasonable'' model of real computation. Indeed, access to such a rounding primitive is known to enable polynomial-time algorithms even for PSPACE-complete problems~\cite{HartmanisS74}. To circumvent this, we avoid constructing $\hat A$ altogether and instead work directly with $A$ and \emph{approximate} solutions of the form $a + b + c = 0 \pm \Delta$. As a consequence, we need approximate versions of all 
necessary additive combinatorics statements (see \cref{sec:toolkit:sec:add-comb}), and an approximate version of subset sum covers and all related lemmas (see \cref{sec:toolkit:sec:subset-sum-cover}). This is mostly tedious, and we omit further details here.

\paragraph{Problem 2: Different Scales}
The more serious obstacle is that there may be no approximation threshold $\Delta$ that simultaneously satisfies Properties~\ref{itm:round-struct} and~\ref{itm:round-false-positives}. For example, consider a set $A = \set{a_1, \dots, a_n}$ whose elements lie at very different scales, say $|a_i| \in [2^i, 2^{i+1}]$. For any choice of $\Delta$, some elements (those with $2^{i+1} \leq \Delta$) will participate in too many false positives, violating~\ref{itm:round-false-positives}, while for others (those with $2^i \geq \Delta$) we have not introduced enough additive structure, violating~\ref{itm:round-struct}.

Our solution is to \emph{manually} separate these scales, i.e., to partition $A$ in a way so that each part contains only elements from roughly the same ``scale''. Within each part we apply the previous machinery to perform the real-to-integer reduction. But we also need to detect solutions across different parts. Here we proceed differently, exploiting that the search space of solutions across different parts is highly restricted. In the example above, every 3SUM solution must be of the form $a_i + a_j + a_{j+1} = 0$, and all such solutions can be found in $O(n)$ time by enumerating~$j$.

Turning this intuition into a formal argument is highly nontrivial and requires substantial effort. In particular, our approach does not yield a direct reduction from Real 3SUM to Integer 3SUM, but instead proceeds via an intermediate problem, namely Real Convolution 3SUM:
\begin{itemize}
\item First, we reduce Real 3SUM to Real Convolution 3SUM using a coarse-grained rounding step. At this stage, the goal is not to tightly control false positives, but rather to identify convolution structure. Crucially, we thereby ensure that the total number of approximate solutions is at most $n^2$ (before, there could have been up to $n^3$ approximate solutions). See \cref{sec:real-to-real-conv} for details, and in particular \cref{lem:reduction-to-clustered}, which is the key lemma underlying scale separation.
\item Second, we reduce Real Convolution 3SUM to Integer 3SUM via a more refined rounding step. Here we really need to keep the number of false positives small: smaller than the trivial upper bound of $n^2$ (in contrast, a direct approach from 3SUM could incur up to $n^3$ false positives now). See \cref{sec:real-conv-to-int} for the details on how to separate scales. Another interesting trick we build on here is an analogue of ``Fredman's trick''~\cite{Fredman76,ChanVX23} for convolutions; see the discussion in \Cref{sec:type3}.
\end{itemize}

%% file: sections/prelims.tex
\section{Preliminaries} \label{sec:prelims}

\subsection{Notation}

Throughout the paper, we use $[n] \coloneqq \{1,2,\ldots,n\}$ and $[\pm n] \coloneqq \set{-n, ..., n}$. For $\Delta \in \Realnneg$, we write $x = y \pm \Delta$ to denote $x \in [y - \Delta, y + \Delta]$. We use the $\tilde O$ notation to suppress polylogarithmic factors, i.e., $\tilde O(T) = O(T \log^{O(1)}T))$. We say that an event occurs \emph{with high probability} if its probability is at least $1 - 1/n^c$ for an arbitrarily large prespecified constant $c$.

\subsection{Problems}

In this section we recall formal definitions of the problems that we consider.

\begin{problem}[Integer 3SUM]
Given a set of $n$ integers $A \subseteq \Int $, decide if there exist three elements $a,b,c \in A$ such that $a+b+c=0$.
\end{problem}

\begin{problem}[Real 3SUM]
Given a set of $n$ real numbers $A \subseteq \Real$, decide if there exist three elements $a,b,c \in A$ such that $a+b+c=0$.
\end{problem}

For the convolution variant of the problem, we find it more convenient to depart from the standard definition, and instead define it for sequences indexed with integers between $-n$ and $n$, so that solutions satisfy a symmetric equation involving indices $i+j+k=0$. Reducing back and forth between this and the usual variant of Real Convolution 3SUM is a simple exercise; see, e.g., \cite{DudekGS20}.

\begin{problem}[Real Convolution 3SUM]
Given a sequence of $2n+1$ real numbers $a\in \Real^{[\pm n]}$, decide if there exist three indices $i,j,k \in [\pm n]$ such that $a[i] + a[j] + a[k] = 0$ and $i + j + k = 0$.
\end{problem}

%% file: sections/toolkit.tex
\section{Real-to-Integer Toolkit} \label{sec:toolkit}

In this section, we develop the tools that we will use throughout our reduction. We first restate the Frank--Tardos theorem and derive its implementation in the comparison-addition model~(see \Cref{sec:toolkit:sec:frank-tardos}). Next, we convert some definitions and lemmas from additive combinatorics into corresponding versions for approximate reals~(see \Cref{sec:toolkit:sec:add-comb}). Then, we define subset sum covers, show how to construct and use them in combination with Frank--Tardos, and prove our ``Poor Man's'' Freiman's theorem~(see \Cref{sec:toolkit:sec:subset-sum-cover}). The final necessary tool for us is linear hashing~(see \Cref{sec:toolkit:sec:hashing}).

\subsection{Frank--Tardos Theorem} \label{sec:toolkit:sec:frank-tardos}
A basic building block of our result is the following theorem due to Frank and Tardos~\cite{FrankT87}.

\thmfranktardos*

\begin{remark}\label{rem:frank-tardos}
Let us add two practical remarks on how to apply \cref{thm:frank-tardos}. First, while the theorem is only stated for vectors $x \in \set{-1, 0, 1}^d$, it can be generalized in a black-box way to $x \in \set{-C, \dots, C}^d$, for any integer constant $C$. Indeed, simply replace each $\alpha_i$ by $O(\log C)$ copies $\alpha_i, 2\alpha_i, 4\alpha_i, \dots, 2^{\ceil{\log C}} \alpha_i$; then clearly all numbers $-C \alpha_i, \dots, C \alpha_i$ can be represented as $\set{-1, 0, 1}$-weighted sums of the newly added items. The transformation only increases $d$ by a factor of $O(\log |C|) = O(1)$ and is therefore negligible.

Second, instead of testing $\alpha^T x \geq 0$ we are typically rather interested in testing $|\alpha^T x| \leq \Delta$, for some threshold $\Delta \in \Realnneg$. To deal with this condition, we apply \cref{thm:frank-tardos} on the vector $\alpha$ with $\Delta$ appended as the last coordinated (and thereby increased dimension from $d$ to $d + 1$). The theorem then yields $v$ and some integer threshold $D$ such that $|\alpha^T x| \leq \Delta$ is equivalent to $|v^T x| \leq D$.
Furthermore, if we add $C$ copies of $\Delta$ to $\alpha$, then all extra coordinates in $v$ must be the same integer $D$ that now acts as a threshold such that $|\alpha^Tx| \le C'\Delta$ if and only if $|v^Tx| \le C'D$ for any $C' \le C$. 
\end{remark}

\subsubsection{Proof}
To demonstrate that the only operations that \cref{thm:frank-tardos} performs on reals are comparisons and additions, we quickly prove how \cref{thm:frank-tardos} follows from known tools not involving reals. We follow the exposition by Eisenbrand~\cite{Eisenbrand10}. Readers familiar with the Frank--Tardos theorem are encouraged to skip this subsection.

The only necessary tool is an algorithmic version of Dirichlet's theorem; see~\cite{Eisenbrand10} for how this follows from an application of the LLL lattice reduction algorithm.

\begin{theorem}[Dirichlet's Theorem for Rationals~{\cite[Lemma 14.9]{Eisenbrand10}}] \label{thm:dirichlet-rat}
There is an algorithm that, given rationals $\alpha \in \Rat^d$ with numerators and denominators bounded by $N$ and an integer parameter $Q \geq 1$, runs in time $(d \log (N Q))^{O(1)}$ and computes integers $q \in \Int$ and $p \in \Int^d$ such that
\begin{enumerate}[label=(\roman*)]
    \item $1 \leq q \leq 2^{d^2} Q^d$, and
    \item $|q \alpha_i - p_i| < 1/Q$ for $i = 1, \dots, d$.
\end{enumerate}
\end{theorem}

From \cref{thm:dirichlet-rat} it is easy to derive the following analogous version for real numbers in the comparison-addition model.

\begin{corollary}[Dirichlet's Theorem for Reals~{\cite[Theorem 14.30]{Eisenbrand10}}] \label{cor:dirichlet-real}
There is an algorithm that, given reals $\alpha \in \Real^d$ and an integer parameter $Q \geq 1$, runs in time $(d \log Q)^{O(1)}$ and computes integers $q \in \Int$ and $p \in \Int^d$ such that
\begin{enumerate}[label=(\roman*)]
    \item $1 \leq q \leq 2^{d^2+d} Q^d$, and
    \item $|q \alpha_i - p_i \beta| < \beta/Q$ for $i = 1, \dots, d$, where we write $\beta = \norm{\alpha}_\infty$.
\end{enumerate}
\end{corollary}
\begin{proof}
The idea is to first round the reals to rationals with sufficiently high precision, and to call \cref{thm:dirichlet-rat}. Specifically, let $N = 2^{d^2+d+1} Q^{d+1}$ and round the given reals to rationals $\alpha_1', \dots, \alpha_d'$ by:
\begin{equation*}
    \alpha_i' = \frac{\floor{\alpha_i / \beta \cdot N}}{N} \in [-1, 1].
\end{equation*}
Even though this definition appears to involve flooring, these rationals can actually each be computed by only $(\log N)^{O(1)}$ comparisons and additions. Specifically, observe that all $\alpha_1', \dots, \alpha_d'$ are rationals with the common denominator $N$ and with numerators in the range $\set{-N, \dots, N}$. Therefore, to determine $\alpha_i'$ we binary-search for its numerator. In each of the $O(\log N)$ search steps, we query whether $\alpha_i' \geq \frac{j}{N}$ (for some $j \in \set{-N, \dots, N}$), which can be equivalently phrased as
\begin{equation*}
    N \alpha_i \leq j \beta.
\end{equation*}
Each such query can be implemented in $O(\log N)$ additions to compute the reals $N \alpha_i$ and $j \beta$ via repeated doubling, followed by one comparison.

After the rounding, we apply \cref{thm:dirichlet-rat} for the rationals $\alpha_1', \dots, \alpha_d'$ with parameter $Q' = 2Q$. Let $q, p_1, \dots, p_d$ denote the resulting integers. On the one hand, the bound on $q$ is immediate. On the other hand,
\begin{equation*}
    |q \alpha_i - p_i \beta| \leq |q (\alpha_i - \alpha_i' \beta)| + |q \alpha_i' \beta - p_i \beta|,
\end{equation*}
where the first term can be bounded by
\begin{equation*}
    |q (\alpha_i - \alpha_i' \beta)| = q \cdot |\alpha_i - \alpha_i' \beta| \leq q \cdot \frac{\beta}{N} \leq \frac{2^{d^2+d} Q^d \cdot \beta}{2^{d^2+d+1} Q^{d+1}} \leq \frac{\beta}{2Q},
\end{equation*}
and the second term can be bounded by $1 / 2Q$ by the guarantee of \cref{thm:dirichlet-rat}.

Finally, focus on the running time. We perform $d (\log N)^{O(1)} = (d \log Q)^{O(1)}$ additions and comparisons in the rounding step, and \cref{thm:dirichlet-rat} takes time $(d \log (N Q))^{O(1)} = (d \log Q)^{O(1)}$.
\end{proof}

\begin{proof}[Proof of \cref{thm:frank-tardos}]
The proof is by repeated applications of Dirichlet's theorem. Specifically, consider the following recursive algorithm. If $\alpha = 0$, then we can trivially return $v = 0$. Otherwise we run the following steps. Write $\beta \coloneqq \norm{\alpha}_\infty$.
\begin{itemize}
    \item Apply \cref{cor:dirichlet-real} on $\alpha$ with parameter $Q = d$ to obtain integers $q \in \Int$ and $p \in \Int^d$.
    \item Recurse on $\hat\alpha \coloneqq q \alpha - \beta p \in \Real^d$ to obtain integers $\hat v \in \Int^d$.
    \item Let $M = \norm{\hat v}_1 + 1$ and return $v = M p + \hat v$.
\end{itemize}
With each recursive call the algorithm indeed makes progress. By definition there is some entry~$i$ with $|\alpha_i| = \beta > 0$. For that entry Dirichlet's theorem is forced to pick $|p_i| = q$ (as otherwise the condition $|q \alpha_i - p_i \beta| < \beta / d$ would be violated). But then by construction $\hat\alpha_i = 0$, i.e., the number of zero coordinates has increased going from $\alpha$ to $\hat\alpha$. In particular, after at most $d$ recursive calls the algorithm stops.

\paragraph{Correctness of Property~(ii)}
Let $x \in \set{-1, 0, 1}^d$ be arbitrary; our goal is to show that $\alpha^T x \geq 0$ if and only if $v^T x \geq 0$. The proof involves two cases:
\begin{enumerate}
    \item \emph{If $p^T x \neq 0$:} First assume $p^T x > 0$. Since $p$ and $x$ are integer vectors, we have $p^T x \geq 1$ and thus $\beta p^T x \geq \beta$. Moreover, Dirichlet's theorem implies that $|(q \alpha - \beta p)^T x| < d \cdot \beta / d = \beta$. Combining both it follows that $q \alpha^T x > 0$ and therefore $\alpha^T x > 0$. By construction we also have
    \begin{equation*}
        v^T x = (M p + \hat v)^T x \geq M - \norm{\hat v}_1 > 0,
    \end{equation*}
    and so the claim holds. In the alternative case that $p^T < 0$, the symmetric proof shows that simultaneously $\alpha^T x < 0$ and $v^T x < 0$.
    \item\emph{If $p^T x = 0$:} In this event, $\alpha^T x \geq 0$ is equivalent to $\hat\alpha^T x \geq 0$. By induction, this is equivalent to $\hat v^T x \geq 0$. By construction, this is further equivalent to $v^T x \geq 0$ (using again that $p^T x = 0$).
\end{enumerate}

\paragraph{Correctness of Property~(i)}
It remains to bound $\norm{v}_\infty$. The entries in $v$ are dominated by the term $M p$. Recall that by Dirichlet's theorem we have $|p| \leq 2^{O(d^2)}$. Therefore,
\begin{equation*}
    \norm{v}_\infty \leq 2^{O(d^2)} \cdot (\norm{\hat v}_1 + 1) \leq 2^{O(d^2)} \cdot (\norm{\hat v}_\infty + 1).
\end{equation*}
As the algorithm reaches recursion depth at most $d$, the final output satisfies $\norm{v}_\infty \leq 2^{O(d^3)}$ as claimed.

\paragraph{Running Time}
In each execution of the algorithm we first apply Dirichlet's theorem in time $(d \log Q)^{O(1)} = d^{O(1)}$, recurse, and then construct $v$ in time $(d \log \norm{v}_\infty)^{O(1)} = d^{O(1)}$. Over the at most $d$ recursive calls, the total time is $d^{O(1)}$.
\end{proof}

\subsubsection{Implementation in the Realistic Comparison-Addition Model} \label{sec:toolkit:sec:frank-tardos:sec:realistic}
Before we have only made sure that the Frank--Tardos theorem can be implemented in the standard comparison-addition model. But, in fact, it even works for the ``realistic'' comparison-addition model introduced before. Recall that in this model the \emph{only} operation that we are allowed to perform on reals are comparisons of the form
\begin{equation*}
    \sum_{i=1}^m n_i \alpha_i \geq 0,
\end{equation*}
for integers $n_i$ and reals $\alpha_i$. Each such operation is assumed to take time $(\sum_{i=1}^m \log |n_i|)^{O(1)}$, which is the same as the time complexity of this comparison if the $\alpha_i$'s were integers.

First observe that the only place where we actually invoke comparisons is in the proof of Dirichlet's theorem for real numbers, \cref{cor:dirichlet-real}, where we binary-search for the numerators of the rational numbers $\alpha_i'$. Each such comparison is of the form $n_1 \alpha_1 + n_2 \alpha_2 \geq 0$, where $|n_1|, |n_2| \leq 2^{O(d^2)}$. In particular, these comparisons run in time $d^{O(1)}$ in the realistic comparison-addition model. The total running time is still polynomial.

However, in the proof of the Frank--Tardos theorem, \cref{thm:frank-tardos}, we also perform additions on reals. Since additions are not supported in the realistic comparison-addition model, instead pretend that each addition is performed \emph{symbolically}. That is, we artificially represent each real number as a list $(n_1, \dots, n_m; \alpha_1, \dots, \alpha_m)$ representing $\sum_{i=1}^m n_i \alpha_i$; here, $n_1, \dots, n_m$ are nonzero integers, and~$\alpha_1, \dots, \alpha_m$ are pointers to the real numbers in the initial instance. This representation takes space $O(\sum_i \log |n_i|)$, and allows for additions in time $(\sum_i \log |n_i|)^{O(1)}$. It can be verified that all reals in the proof of \cref{thm:frank-tardos} have representation size at most $\sum_i \log |n_i| \leq d^{O(1)}$ (the length is $m \leq d$ as there are only $d$ input reals, and each integer coefficient $n_i$ is bounded by $\log |n_i| \leq d^{O(1)}$ throughout). Therefore, the overhead due to the realistic comparison-addition model is $d^{O(1)}$, and thus the total time is still polynomial.

\subsection{Basic Additive Combinatorics for Approximate Reals} \label{sec:toolkit:sec:add-comb}

Let $A, B \subseteq \Real$ and $\Delta \in \Realnneg$. Recall that we write $x = y \pm \Delta$ to denote $x \in [y-\Delta, y+\Delta]$.
For a threshold $t \in \Intpos$, we say that $A$ is \emph{$(\Delta, t)$-spaced} if, for every $a \in A$, there are at most $t$ elements of $A$ within the distance $\Delta$ from $a$, i.e.,
\begin{equation*}
    \abs{\set{b \in A : a = b \pm \Delta}} \le t.
\end{equation*} 
Whenever we work with spaced sets, we only care about sets with a constant number of elements per width-$\Delta$ interval, so $t = O(1)$. 
We also just write \emph{$\Delta$-spaced} to refer to $(\Delta, 1)$-spaced.

We define the \emph{$\Delta$-approximate energy} of $A$ and $B$ to be 
\begin{equation*}
    \enerapx(A, B) \coloneqq \abs{\set{(a_1, a_2, b_1, b_2) \in A^2\times B^2 : a_1 + b_1 = a_2 + b_2 \apx}}.
\end{equation*}
For a single set, we write $\enerapx(A) = \enerapx(A,A)$.

Also denote the number of \emph{$\Delta$-approximate representatives} of $x \in \Real$ in $A+B$ by \begin{equation*}
     \reprapx_{A+B}(x) \coloneqq |\{(a, b) \in A \times B : a + b = x \apx\}|
\end{equation*} and similarly for $A-B$ or more than two sets.
Notice that we can relate the two definitions by \begin{equation*}
     \enerapx(A, B) = \sum_{a \in A}\sum_{b \in B} \reprapx_{A+B}(a + b) = \sum_{a_1 \in A}\sum_{a_2 \in A} \reprapx_{B-B}(a_1 - a_2).
\end{equation*}
The upcoming lemmas shows how $\Delta$-approximate energy bounds the number of 3SUM solutions up to $\pm \Delta$ for $(\Delta, O(1))$-spaced sets.

\begin{lemma}\label{lem:low-energy-few-solutions}
Let $\Delta \in \Realnneg$ and let $A, B, C \subseteq \Real$ be $(\Delta, O(1))$-spaced. Then the number of solutions $(a, b, c) \in A \times B \times C$ to $a + b + c = 0 \pm \Delta$ is at most
\begin{equation*}
    O\left(\sqrt{E^{\pm \Delta}(A,B) |C|}\right).
\end{equation*}
\end{lemma}
\begin{proof}
    First, observe that the number of solutions corresponds exactly to
    \begin{equation*}
        \sum_{c \in C} r^{\pm \Delta}_{A+B}(-c).
    \end{equation*}
    Since $C$ is $(\Delta, O(1))$-spaced, notice that every interval of width $2\Delta$ contains at most constantly numbers in $C$. Therefore, we can roughly represent the size of $C$ as 
    \begin{equation*}
         2|C| = \sum_{z \in \Int} |[(z-1)\Delta, (z+1)\Delta] \cap C| = O\parens*{\sum_{z \in \Int} |[(z-1)\Delta, (z+1)\Delta] \cap C|^2}.
    \end{equation*}
    Further notice that for $z \in \Int$, one of $[(z-1)\Delta, z\Delta]$ and $[z\Delta, (z+1)\Delta]$ contains at least $r^{\apx}_{A+B}(z\Delta) / 2$ many sums $a+b$, which all lie in an interval of $\Delta$. Hence, these sums contribute at least $r^{\apx}_{A+B}(z\Delta)^2 / 4$ to the $\Delta$-approximate energy. Summing over all $z \in \Int$, this width-$\Delta$ interval can be counted at most twice, giving
    \begin{equation*}
        \sum_{z \in \Int} r^{\apx}_{A+B}(z\Delta)^2 \le \frac{1}{8}E^{\apx}(A,B).
    \end{equation*}
    Using Cauchy--Schwarz, the number of solutions to the equation can then be bounded by 
    \begin{align*}
        \sum_{c \in C} r^{\apx}_{A+B}(-c)
        &\le \sum_{z \in \Int} |[(z-1)\Delta, (z+1)\Delta] \cap C|\cdot r^{\apx}_{A+B}(z\Delta)\\
        &\le \parens*{\sum_{z \in \Int} |[(z-1)\Delta, (z+1)\Delta] \cap C|^2}^{1/2}\parens*{\sum_{z \in \Int} r^{\apx}_{A+B}(z\Delta)^2}^{1/2}\\
        &\le O\parens*{\sqrt{|C|E^{\apx}(A,B)}}.\qedhere
    \end{align*}
\end{proof}

The next lemma gives us a tool to extract a subset that contains most of the energy quadruples of a set. Additionally, notice that this decomposition not only gives us a subset $P \subseteq A$ whose elements contribute a large amount to the total approximate energy. It also guarantees that \emph{every} element in $P$ is in many approximate energy solutions with \emph{only} other elements from $P$.

\begin{lemma}[Popular/Unpopular Decomposition] \label{lem:decomp-pop}
Let $A \subseteq \Real$ and let $K \geq 1$. There is a partition $A = P \sqcup U$ into \emph{popular} and \emph{unpopular} elements such that:
\begin{enumerate}[label=(\roman*)]
    \item For each $a \in P$ we have $r_{P+P-P}^{\pm \Delta}(a) \geq |A|^2 / 4K$, and
    \item $E^{\pm \Delta}(A, U) < |A|^3 / K$.
\end{enumerate}
The sets $P$ and $U$ can be computed in randomized time $\tilde O(|A| K^2)$ or deterministic time $\tilde O(|A|^2)$.
\end{lemma}
\begin{proof}
We start with a description of the simpler (but slower) deterministic algorithm, and then comment on the randomized algorithm later.

\paragraph{Deterministic Algorithm}
We initialize $P \coloneqq A$ and $U \coloneqq \emptyset$. While there exists some element $a \in P$ with $r_{P+P-P}^{\pm \Delta}(a) < |A|^2 / 2K$, move $a$ from $P$ to $U$. Once this process stops, we claim that $P$ and $U$ are as stated.

Indeed, the first property holds directly by construction. For the second property, imagine that we keep a counter of the additive quadruples $(a_1, a_2, a_3, a_4) \in A^4$ with $a_1 - a_2 = a_3 - a_4 \pm \Delta$, where~$a_3$ or~$a_4$ is part of $U$. Initially that counter is zero as $U$ is empty. Whenever we move some item $a$ from~$P$ to~$U$, the only new additive quadruples are those involving~$a$ and three other elements from~$P$; any other quadruple is already accounted for. However, the number of these additive quadruples is at most $2r_{P+P-P}^{\pm \Delta}(a) < |A|^2 / K$ (where the factor $2$ is due to the fact that~$a$ may appear in either $a_3$ or $a_4$). In total, over the at most $|A|$ iterations, the counter reaches less than~$|A|^3 / K$, and the bound on $E^{\pm \Delta}(A, U)$ follows.

This algorithm can be implemented in deterministic time $\tilde O(|A|^2)$ by maintaining a data structure that supports range-queries and range-updates of $r_{P-P}^{\pm \Delta}(a)$ in polylogarithmic time.

\paragraph{Randomized Algorithm} 
To achieve the close-to-linear randomized running time, we cannot afford to determine the multiplicities $r_{P+P-P}^{\pm  \Delta}(a)$ exactly. Instead, we rely on sampling to get sufficiently precise estimates. Specifically, by sampling $\Theta(K \log |A|)$ uniformly random pairs \makebox{$a_1, a_2 \in P$} and counting for how many of the sampled pairs there exists $a_3 \in P$ with $a = a_1 + a_2 - a_3 \pm \Delta$, we can interpolate an additive approximation of $r_{P+P-P}^{\pm  \Delta}(a)$ within error $|P|^2 / 16K \leq |A|^2 / 16K$, say. It succeeds with high probability by Chernoff's bound.

Equipped with this tool, the full randomized algorithm is as follows. It runs in \emph{stages}. In each stage, we enumerate all $a \in P$, and approximate $r^{\pm \Delta}_{P+P-P}(a)$ with additive error $|A|^2 / 16K$ by the sampling procedure from before. If that approximation falls below the threshold $\tau = |A|^2 / 2K - |A|^2 / 16K$, we move $a$ from $P$ to~$U$ (as in that case we have $r_{P+P-P}^{\pm \Delta}(a) < |A|^2 / 2K$, just like in the deterministic algorithm). If in this stage the algorithm moved at least $|A| / 8K$ elements to~$U$, then we proceed to the next stage and repeat. Otherwise we terminate.

We claim that once this algorithm terminates, the sets $P$ and $U$ are as claimed. On the one hand, the second property can be proved exactly as before, given that each element $a$ that moves from~$P$ to~$U$ satisfies $r_{P+P-P}^{\pm \Delta}(a) < |A|^2 / 2K$. On the other hand, to see that the first property holds, consider the final stage of the algorithm. For each remaining item $a$, at the time when the algorithm had reached~$a$ it determined that $r_{P+P-P}^{\pm \Delta}(a) \geq \tau - |A|^2 / 16K = |A|^2 / 2K - |A|^2 / 8K$. Afterwards, it moved less than $|A| / 8K$ items to $U$ (since this is the final stage), and therefore once the algorithm terminates we have \makebox{$r_{P+P-P}^{\pm \Delta}(a) \geq |A|^2 / 2K - |A|^2 / 8K - |A|^2 / 8K = |A|^2 / 4K$}.

There can be at most $8 K$ stages, and so the total running time is $O(K \cdot |A| \cdot K \log |A|) = \tilde O(|A| K^2)$ as claimed.
\end{proof}

\subsection{Subset Sum Covers} \label{sec:toolkit:sec:subset-sum-cover}

\begin{definition}[Subset Sum Cover] \label{def:subset-sum-cover}
Let $A \subseteq \Real$. A pair $(\alpha, f)$, where $\alpha \in \Real^d$ and $f : A \to \set{0, 1}^d$, is called a \emph{$\Delta$-approximate subset sum cover} of $A$ if it satisfies $a = \alpha^T f(a) \pm \Delta$ for each $a \in A$. We call $d$ the \emph{rank} of the subset sum cover.
\end{definition}

We remark that if we aim for better approximation than $\pm \Delta$, say a $(\Delta / C)$-approximate subset sum cover for a constant $C$, this follows in a black-box way. Indeed, we can simply append $\Delta, 2\Delta, \dots, 2^{\lceil \log C \rceil}\Delta$ to $\alpha$ and binary search the correct coefficient for each $a \in A$. This adds only $O(\log C)$ dimensions and $O(|A|\log C)$ to the running time.
We note that this flexibility in the choice of $\Delta$ is a key detail for the main reduction later.

\begin{observation}\label{obs:refine-cover}
    For a constant $C\in \Intpos$, we can refine a $\Delta$-approximate subset sum cover of $A$ into a $\Delta / C$-approximate subset sum cover, while only increasing the rank by a constant factor.
\end{observation}

\begin{lemma}[Subset Sum Cover for Popular Elements] \label{lem:subset-sum-cover}
Let $\Delta \in \Realnneg$ and let $A \subseteq \Real$ be a $(\Delta, O(1))$-spaced set. Let $K \ge 1$ and assume that for all $a \in A$ it holds that \smash{$r_{A+A-A}^{\pm \Delta}(a) \geq |A|^2 / K$}. Then there is a $\Delta$-approximate subset sum cover of~$A$ with rank $O(K \log^2 |A|)$. It can be computed in randomized time $\tilde O(|A| K)$ or deterministic time $\tilde O(|A|^2)$.
\end{lemma}
\begin{proof}
For sets $A$ and $D \subseteq A - A$, let $\Gamma^{\pm \Delta}_A(D)$ denote the undirected graph on the vertex set $A$ that contains an edge $\set{a, b}$ if and only if there is $d \in D \cup (-D)$ with $a - b = d \pm \Delta$. (In case that $A$ is a group and $\Delta = 0$, this graph is essentially the Cayley graph of $A$ with respect to $D$.)

The proof works by constructing a set $D\subseteq A-A$ of size $d = O(K\log^2|A|)$ such that $\Gamma^{\pm \Delta}_A(D)$ has only $O(K)$ components. Once we found such a set $D$, we can construct the $\Delta$-approximate subset sum cover as follows.

Explicitly construct $\Gamma^{\pm \Delta}_A(D)$, and compute all connected components $C_1, \dots, C_\ell$ where $\ell = O(K)$. In order to construct $\alpha \in \Real^{\ell+d}$, first take an arbitrary node in each component $C_i$ and assign it to $\alpha_i$. Then, set $\alpha_{\ell+1}, \dots, \alpha_{\ell+d}$ to the elements of $D$. We now show that any $a \in A$ can be written as $a = \alpha^Tf(a) \pm |A|\Delta$ for some function $f : A \to \set{-|A|, \dots, |A|}$.
Let $a \in A$ and suppose $a \in C_i$. Let $(\alpha_i = b_1, b_2, \dots, b_p = a)$ denote a shortest path from $\alpha_i$ to $a$. Then, we can express $a$ as the telescopic sum
\begin{equation*}
    a = \alpha_i + \sum_{j=1}^{p-1} (b_{j+1} - b_j) \pm p\Delta,
\end{equation*}
where for each term we have $b_{i+1} - b_i = d \pm \Delta$ for some $d \in D \cup (-D)$, since $\set{b_i, b_{i+1}}$ is an edge in $\Gamma^{\pm  \Delta}_A(D)$. Thus, $a$ can be expressed as an approximate linear combination $a = \alpha^T f(a) \pm |A|\Delta$ with $f(a) \in \{-|A|, \dots, |A|\}$, since the sum contains only $p \le |A|$ terms.
To construct $f$ from $D$, we can run a BFS from each $\alpha_i$ for $i \in [\ell]$ to find the shortest path. By constructing representations incrementally from the representation of the predecessor, the running time is bounded by the number of edges in $\Gamma^{\pm \Delta}_A(D)$, so $\tilde O(|A|K)$. To construct $\Gamma^{\pm\Delta}_A(D)$ in the same time, simply sort $A$ increasingly and enumerate every pair $(a,d) \in A \times (D \cup (-D))$. For each such pair, binary search for all $b \in A$ with $a - b = d \pm \Delta$; since $A$ is $(\Delta, O(1))$-spaced, there are only constantly many.

Finally, we can turn $(\alpha, f)$ into a $\Delta$-approximate subset sum cover of $A$ as follows. We add an additional coordinate $\Delta$ to $\alpha$ to replace $\pm |A|\Delta$ by $\pm \Delta$. We add negations of each $\alpha_i$ and then replace each $\alpha_i$ by all of $\alpha_i, 2\alpha_i, 4\alpha_i, \ldots, 2^{\lceil\log |A|\rceil}\alpha_i$. This way, we can guarantee that all coefficients are in~$\set{-1, 0, 1}$.

We now state the randomized algorithm to construct $D$, and then show how to derandomize it.

\paragraph{Randomized Algorithm}
The set $D$ consists of $O(K\log^2|A|)$ elements. We create each element in $D$ independently, by choosing two elements $a$ and $b$ uniformly at random from $A$ and including $a-b$ in $D$.
For the analysis, think of the construction of $D$ in $\log |A|$ \emph{stages}. Denote by $D_i$ the edge set at the beginning of stage $i$. In this stage, we generate $O(K\log|A|)$ differences to add to $D$.

We prove that after stage $i$ every connected component of $\Gamma^{\pm \Delta}_{A}(D_{i+1})$ must have size at least $\min\set{2^i, |A|/2tK}$, where $t$ is the parameter such that $A$ is $(\Delta, t)$-spaced.
Denote by $C_1, \dots, C_\ell$ the connected components of $\Gamma^{\pm \Delta}_A(D_i)$ of size less than $\min\set{2^i, |A| / 2tK}$. For each component, fix an arbitrary $a_i \in C_i$. Note that for any $i \in [\ell]$, we have
\begin{align*}
    \sum_{b \in A \setminus C_i} r^{\pm \Delta}_{A - A}(a_i - b) &= r^{\pm \Delta}_{A+A-A}(a_i) - \sum_{b \in C_i} r^{\pm \Delta}_{A - A}(a_i - b) \\
    &\ge \frac{|A|^2}{K} - t|C_i| \, |A|\\
    &\ge \frac{|A|^2}{2K}.
\end{align*}
Thus, the probability that we sample one of the differences in $a_i - (A \setminus C_i)$ in stage $i$ is at least
\begin{equation*}
    1 - \parens*{1 - \frac{1}{2K}}^{O(K \log |A|)} \geq 1 - 1/\poly(|A|).
\end{equation*}
By union bound, the error probability over all components is at most $|A| \cdot 1/\poly(|A|) = 1/\poly(|A|)$. Hence, with high probability, each $C_i$ is newly connected to a different component of $\Gamma^{\pm \Delta}_{A}(D_i)$. If one of the components has size at least $|A|/2tK$, the new component does as well. Otherwise both components have size at least $2^{i-1}$ and the new component has size at least $2^i$. This proves the induction step. 

By a union bound, each component has size at least $|A|/2tK$ after stage $\log|A|$ with high probability. Hence there can be at most $O(K)$ components.

\paragraph{Deterministic Algorithm}

In the deterministic algorithm, we explicitly operate in $\log|A|$ stages.
In stage $i$, instead of randomly sampling $D_{i+1} \setminus D_i$, we compute all $C_i$ and fix an arbitrary $a_i \in C_i$. Then, we generate a hitting set instance $\set{a_i-b, b-a_i : b \in A \setminus C_i}_{i \in [\ell]}$ and greedily compute a $O(\log|A|)$-approximate solution $S$ to this instance in time $\tilde O(|A|^2)$~\cite[Chapter 2]{Vazirani01}. By the analysis of the randomized algorithm, there is a solution of size at most $O(K\log|A|)$, so $S = O(K\log^2|A|)$. The remaining argument is the same. Constructing and solving these hitting set instances adds $\tilde O(|A|^2)$ to the running time.
\end{proof}

The following lemma combines \Cref{thm:frank-tardos,lem:subset-sum-cover} to compute a hash function that hashes a set of real number in which every number is popular to integers. 
It is the single interface via which our two reductions make use of subset sum covers and the Frank--Tardos theorem.

\begin{lemma}[Hashing Popular Elements to Integers]\label{lem:hash-popular-to-int}
Let $\Delta \in \Realnneg$, let $C \geq 1$ be any constant and let~$A$ be a $(\Delta, O(1))$-spaced set such that each $a \in A$ satisfies
\begin{equation*}
    r_{A+A-A}^{\pm C\Delta}(a) \geq \frac{|A|^2}{K}.
\end{equation*}
There is a randomized algorithm that computes a function $h : A \to [\pm 2^{(K \log |A|)^{O(1)}}]$ such that:
\begin{enumerate}[label=(\roman*)]
    \item For all $a, b, c \in A$ with $a + b + c = 0 \pm \Delta$, we have $\Pr[h(a) + h(b) + h(c) = 0] \geq \Omega(1)$.
    \item For all $a, b, c \in A$ we have  $h(a) + h(b) + h(c) = 0$, we have $a + b + c = 0 \pm 100 \Delta$.
    \item For all $a \in A$, $|\set{b \in A : h(a) = h(b)}| \leq O(1)$.
\end{enumerate}
The algorithm runs in time $\tilde O(|A|K + (K\log|A|)^{O(1)})$.
\end{lemma}
\begin{proof} 
    We first run the algorithm from \Cref{lem:subset-sum-cover} to compute a $C\Delta$-approximate subset sum cover of $A$ with rank $d \le O(K\log^2 |A|)$ in time $\tilde O(|A|K)$. Using \Cref{obs:refine-cover}, we refine that cover into a $\Delta$-approximate subset sum cover $(\alpha, f)$. Next, we run \Cref{thm:frank-tardos} on $\alpha$ to compute $v \in \Int^d$ and $D \in \Intpos$ with $D \le 2^{O(d^3)}$ such that $|\alpha^Tx| \le C\Delta$ if and only if $|v^Tx| \le CD$ for $C \le 100$ and $x \in \{-3, \dots, 3\}^d$, as discussed in \Cref{rem:frank-tardos}.
    
    We round all values of $v^Tf(a)$ to multiples of $10D$ by defining $g(a) \coloneqq \lfloor v^Tf(a) / 10D\rfloor$ for all $a \in A$. Finally, define $h(a) \coloneqq g(a) + c_a$ where $c_a$ is a uniformly random value from $\set{-2, \dots, 2}$. This way $h$ is indeed a function $A \to [\pm 2^{(K\log |A|)^{O(1)}}]$. We prove the three claimed properties separately.
    
    \paragraph{Property (i): Completeness} Let $a,b,c \in A$ with $a + b + c = 0 \pm \Delta$. As each representation $\alpha^Tf(x)$ introduces an additional $\pm \Delta$, we get $|\alpha^T(f(a) + f(b) + f(c))| \le 4 \Delta$. Replacing $\alpha$ with $v$ yields $|v^T(f(a) + f(b) + f(c))| \le 4D$ and $v^Tf(a) + v^Tf(b) = - v^Tf(c) \pm 4D$. After rounding, we get both
    \begin{equation*}
        g(a) + g(b) = \floor*{\frac{v^Tf(a)}{10D}} + \floor*{\frac{v^Tf(b)}{10D}} \le \floor*{\frac{v^Tf(a) + v^Tf(b)}{10D}} \le \floor*{\frac{-v^Tf(c) + 4D}{10D}} \le -g(c) + 1
    \end{equation*}
    and 
    \begin{equation*}
        g(a) + g(b) = \floor*{\frac{v^Tf(a)}{10D}} + \floor*{\frac{v^Tf(b)}{10D}} \ge \floor*{\frac{v^Tf(a) + v^Tf(b)}{10D}}-1 \ge \floor*{\frac{-v^Tf(c) - 4D}{10D}} - 1 \ge -g(c) - 2.
    \end{equation*}
    In total, we know that $g(a) + g(b) + g(c) = 0 \pm 2$. With probability at least $\Omega(1)$, the correct offset is added to $a$, $b$, and $c$ (i.e., $c_a = c_b = 0$ and $c_c = 0 \pm 2$) such that they satisfy $h(a) + h(b) = h(c)$.
    
    \paragraph{Property (ii): Soundness} Let $a,b,c \in A$ with $h(a) + h(b) + h(c) = 0$ and thus $g(a) + g(b) + g(c) = 0 \pm 6$. This implies $|v^T(f(a) + f(b) + f(c))| \le 90D$ and $|\alpha^T(f(a) + f(b) + f(c))| \le 90 \Delta$. Unpacking the subset sum cover yields $a + b + c = 0 \pm 93\Delta$.
    
    \paragraph{Property (iii): Constant Bucket Load} Notice that if $g(a) = g(b)$ for $a,b \in A$, then $|v^T(f(a) - f(b))| \le 10D$ and hence $|\alpha^T(f(a) - f(b))| \le 10\Delta$. Since $A$ is $(\Delta, O(1))$-spaced, there can be at most $O(1)$ such elements $b$. In other words, for every $x$ there are at most constantly many $a \in A$ with $g(a) = x$. Furthermore, if $h(a) = h(b)$, then $c_a = c_b \pm 4$ and $g(a) = g(b) \pm 4$. Since there are only $O(1)$ elements $b$ with $g(b) = x$ for each of the 9 options $x = g(a) \pm 4$, property (iii) follows. 

    \paragraph{Running Time} \Cref{lem:subset-sum-cover} runs in time $\tilde O(|A|K)$ and returns a subset sum cover of rank $O(K \log^2|A|)$. Invoking \Cref{thm:frank-tardos} takes time $(K \log |A|)^{O(1)}$ and returns integers in the desired range. All other operations take place on integers with $(K\log|A|)^{O(1)}$ bits and hence do not change the running time.
\end{proof}

Using the decomposition from \Cref{lem:decomp-pop} and subset sum covers from \Cref{lem:subset-sum-cover}, we can immediately prove our ``Poor Man's'' Freiman's theorem---even in an approximate version that is coherent with the previous notation.
Note that setting $\Delta = 0$ directly gives an exact subset sum cover of the same rank for any set $A \subseteq \Real$ with doubling $|A+A|/|A| \le K$.

\begin{theorem}[Poor Man's Freiman's Theorem]\label{thm:poor-man-freiman}
    Let $\Delta \in \Realnneg$, $K \ge 1$, and let $A$ be a $\Delta$-spaced set such that $A+A$ can be covered with $K|A|$ intervals of width $\Delta$.
    Then there is a $\Delta$-approximate subset sum cover of $A$ with rank $(K \log|A|)^{O(1)}$. 
\end{theorem}
\begin{proof}
    Denote by $\mathcal{I}$ the set of at most $K|A|$ intervals covering $A+A$.
    The proof consists of three steps. First, we show that the ``small approximate doubling'' condition $|\mathcal{I}| \le K|A|$ implies large approximate energy of $A$, i.e., $E^{\pm \Delta}(A) \ge |A|^3 / K$.
    Secondly, we find a large subset $P \subseteq A$ that can be covered by a small subset sum cover using our tools in \Cref{lem:decomp-pop,lem:subset-sum-cover}.
    Finally, we show how to cover the full set $A$ using a small set of translates of $P-P$, which resembles an approximate version of Ruzsa's covering lemma.
    
    Since every pair of sums $a+b, c+d \in I$ for some $I \in \mathcal{I}$ and $a,b,c,d \in A$ satisfies $a+b = c+d \pm \Delta$, it contributes to the $\Delta$-approximate energy. Using Cauchy--Schwarz, we get
    \begin{align*}
        E^{\pm\Delta}(A)
        &\ge \sum_{I \in \mathcal{I}} \abs{\set{(a,b) \in A^2 : a+b \in I}}^2 \\
        &\ge \frac{1}{K|A|}\parens*{\sum_{I \in \mathcal{I}} \abs{\set{(a,b) \in A^2 : a+b \in I}}}^2 \\
        &= \frac{|A|^4}{K|A|} \\
        &= \frac{|A|^3}{K}.
    \end{align*}
    
    We apply \cref{lem:decomp-pop} to $A$ with parameter $K$, and obtain a decomposition $A = P \sqcup U$ with $E^{\pm \Delta}(A,U) < |A|^3/K$ and $r^{\pm\Delta}_{P+P-P}(a) \ge |A|^2/4K$. It cannot be the case that $U = A$, as in that case $E^{\pm\Delta}(A) = E^{\pm\Delta}(A, U) < |A|^3 / K$ is a contradiction. Hence, $P$ is non-empty and contains some element $a$. That element, however, satisfies $r^{\pm\Delta}_{P+P-P}(a) \geq |A|^2 / 4K$ and thus witnesses that $|P| \geq \Omega(|A| / \sqrt{K})$. Indeed, for any $a$, we have $r^{\pm\Delta}_{P+P-P}(a) \leq O(|P|^2)$ since $P$ is $\Delta$-spaced.
    \Cref{lem:subset-sum-cover} states that $P$ admits a $\Delta$-approximate subset sum cover $(\alpha, f)$ of rank $O(K \log^2 |A|)$. For the remaining proof, we argue how to modify $(\alpha, f)$ to cover the whole of $A$ without increasing the rank by too much. 
    
    Let $T \subseteq A$ be a maximal subset such that the multiset $\set{t + p : t \in T, p \in P}$ is still $\Delta$-spaced. Then, for every interval $I \in \mathcal{I}$ there can be at most one pair $(t,p) \in T\times P$ with $t + p \in I$. Therefore, we get 
    \begin{equation*}
        |T|\cdot |P| \le c\abs{\mathcal{I}} \le cK|A| \le O(K^{3/2}|P|).
    \end{equation*}
    In particular, $|T| \le O(K^{3/2})$. By the maximality of $T$, for every $a \in A$, there is some $t \in T$ and $p, p' \in P$ such that $t + p = a + p' \pm \Delta$. Hence, $a = p - p' + t\pm \Delta$ and $a = \alpha^Tf(p) - \alpha^Tf(p') + t \pm 3\Delta$ since representing $p$ and $p'$ as a subset both induces an additional $\pm \Delta$.
    By appending the negation of every coordinate to $\alpha$, as well as all $t \in T$ and two copies of $\Delta$, we can adapt the subset sum cover to cover every $a \in A$. This only doubles the rank and adds $O(K^{3/2})$ coordinates.
\end{proof}

By simply following the proof of \Cref{thm:poor-man-freiman}, one can also implement an algorithm that computes the approximate subset sum cover in time $\tilde O(|A|^2)$.

\subsection{Linear Hashing}\label{sec:toolkit:sec:hashing}

Following many previous 3SUM-related results, we also use the following family of almost-additive hash functions due to Dietzfelbinger~\cite{Dietzfelbinger96}. See also~\cite{KopelowitzPP16} for a formal proof of almost-additivity.

\begin{lemma}[Linear Hashing~\cite{Dietzfelbinger96}]\label{lem:linear-hashing}
Let $m \leq n$. There is a family of hash functions $h : [\pm 2n] \to [m]$ and a constant-sized set $\Lambda \subseteq [\pm m]$ with the following properties:
\begin{itemize}
    \item \emph{(Efficiency)} Sampling and evaluating $h$ takes time $(\log n)^{O(1)}$, 
    \item \emph{(Almost-Additivity)} For all keys $x, y \in [\pm n]$, $h(x) + h(y) = h(x + y) + \lambda$, for some $\lambda \in \Lambda$,
    \item \emph{(Pairwise Independence)} For all distinct keys $x, y \in [\pm n]$ and all $i, j \in [m]$,
    \begin{equation*}
        \Pr[\text{$h(x) = i$ and $h(y) = j$}] \leq O\parens*{\frac{1}{m^2}}.
    \end{equation*}
\end{itemize}
\end{lemma}

%% file: sections/real-to-real-conv.tex
\section{Reduction from Real 3SUM to Real Convolution 3SUM} \label{sec:real-to-real-conv}

In this section, we prove the first step of our real-to-integer reduction, namely the following theorem.

\begin{theorem} \label{thm:real-to-real-conv}
    If Real Convolution 3SUM can be solved in time $O(n^{2-\epsilon})$, for some constant $\epsilon > 0$, then Real 3SUM can be solved in time $O(n^{2 - \Omega(\epsilon^2)})$.
\end{theorem}

We prove \Cref{thm:real-to-real-conv} in two steps. The first step resembles a rounding of our $n$ real numbers to $g \le n$ real numbers such that solutions are approximately preserved. This is formalized in the following definition. Think of $g = n^{1-\Theta(\epsilon)}$ and $s = n^{\Theta(\epsilon)}$, as this will be what we choose later.

\newcommand{\model}{\phi}
\begin{definition}\label{def:modeled}
Let $A \subseteq \Real$. For $\model:A\to\Real$ and $\Delta \in \Realnneg$, we call $(\model, \Delta)$ a \emph{$(g, s)$-modeling} of $A$ if
\begin{enumerate}[label=(\roman*)]
    \item $\abs{\model(A)} \leq g$ and $\model(A)$ is $(\Delta, O(1))$-spaced,
    \item $\abs{\model^{-1}(x)} \le s$ for all $x \in \Real$,
    \item any $a,b,c \in A$ with $a + b + c = 0$ satisfy $\model(a) + \model(b) + \model(c) = 0 \pm \Delta$.
\end{enumerate}
We say that a Real 3SUM instance $A$ is \emph{$(g,s)$-modeled} if we have access to a $(g, s)$-modeling of $A$.
\end{definition}

In the upcoming lemma, we reduce any Real 3SUM instance to a small number of modeled instances.
Notice that rounding numbers to integer multiples of $\Delta$ for any choice of $\Delta$ would already preserve all 3SUM solutions up to $\pm 2\Delta$.
However, we cannot actually round real numbers in our limited model of computation.
We show that with appropriate choices of $\Delta$ and a small overhead in the number of instances and $g$, we can construct functions $\phi : A \to A$ that satisfy the properties from \Cref{def:modeled}.
Furthermore, every $b \in A$ satisfies $\phi(b) = a$ if and only if $b \in [a, a + \Delta/2]$. This way, property (iii) is satisfied by construction.

To simultaneously satisfy properties (i) and (ii), we also need to separate different scales present in $A$, as otherwise, many numbers might be rounded to the same element of $A$.
Our reduction hence works in stages, where in each stage we detect one scale with many elements for which rounding to elements of $A$ satisfies properties (i) and (ii).
We define scales in such a way that detecting solutions across scales can be done in subquadratic time.

\begin{lemma}[Reduction to Modeled Instances] \label{lem:reduction-to-clustered}
    Let $\gamma, \delta \in (0,1]$. Then any Real 3SUM instance $A$ can be reduced in time $\tilde O(n^{2 - \delta\gamma/2})$ to at most $n^{3\gamma}$ instances that are $(9n^{1 - \delta/2},n^\delta)$-modeled.
\end{lemma}

\begin{proof}
    Our algorithm proceeds in rounds, where after each round, we are guaranteed to remove at least $s \coloneqq \lceil n^\delta\rceil$ items from $A$. Hence, there are at most $n^{1- \delta}$ rounds, and we can spend linear time per round.
    
    We first sort the instance $A = \{a_1 < a_2 < \ldots\}$ in increasing order. Among all intervals of the form $[a_j, a_{j + s}]$, pick the interval $I$ with the minimum width $\Delta \coloneqq a_{j + s} - a_j$.
    We will remove at least all elements in $I$, possibly more.
    For an element $a$, we write $N(a) \coloneqq |A \cap [a-\Delta, a+\Delta]|$ for the number of elements within distance $\Delta$ to $a$.
    
    For any threshold $t$, we can construct a collection of disjoint width-$\Delta$ intervals $I_1, \ldots, I_\ell$ of width $\Delta$ such that they cover at least $\ell t / 9$ elements in total, while guaranteeing that each element $a$ with $N(a) \ge t$ is covered. For instance, iterate over $A$ in increasing order and greedily create a new interval starting at $a_i$ if it covers an element in $N(a)$ for some $a$ with $N(a) \ge t$. Three consecutive intervals cover the whole neighborhood, so out of three consecutive intervals, at least one must cover $t / 3$ elements.

    Construct such intervals for several thresholds $t_0 \ge t_1 \ge t_2 \ge \ldots$ defined by $t_k \coloneqq n^{\delta(1 - \gamma k /2)}$. Let $\mathcal{I}_k$ denote the set of intervals for $t_k$ and let $A_k$ be all elements of $A$ that are covered by $\mathcal{I}_k$. Denote by $S_k$ the set of all starting points of these intervals $\mathcal{I}_k$, i.e., $S_k \coloneqq \set{\min (I \cap A_k) \mid I \in \mathcal{I}_k}$.

    Now there must be $0 \le k < \ceil{1 / \gamma}$ with $|A_{k+1}| \le n^\gamma |A_k|$, otherwise $|A_{\ceil{1 / \gamma}}| > n^{\gamma \ceil{1/\gamma}} \ge n$. We wan to eliminate all elements in $A_k$ in this round, which includes all elements in $I$.
    To remove all elements in $A_k$, we need to detect all solutions $a + b + c = 0$ with $a \in A_k$. There are two cases:

    \begin{enumerate}
        
        \item\emph{If $b \not\in A_{k-1}$ or $c \not\in A_{k+1}$:} We assume $b \notin A_{k+1}$ since the case $c \notin A_{k+1}$ is symmetric.
        We enumerate all $c \in A$ and all intervals $J \in \mathcal{I}_k$. To form a solution with $a \in J$, the last element~$b$ needs to be in $-(c + J)$. This range necessarily contains less than $t_{k+1}$ elements---otherwise for all $b \in -(c + J) \cap A$ we have $N(b) \ge t_{k+1}$ and it would be covered in $A_{k+1}$. Hence, we can test all elements in this range in time $\tilde O(|A|\cdot|S_k|\cdot t_{k+1} )$.
    
        As we are about to eliminate $|A_k| \ge |S_k| t_k / 9$ elements in this iteration, per eliminated element this runs in time 
        \begin{equation*}
            \tilde O\parens*{\frac{n \cdot |S_k|t_{k+1} }{ |S_k|t_k / 9}} = \tilde O\parens*{n^{1 - \delta\gamma/2}}.
        \end{equation*}
    
        \item\emph{If $b, c \in A_{k+1}$:}
        As $a \in A_k \subseteq A_{k+1}$, it suffices to check for solutions completely contained in~$A_{k+1}$. We distinguish two subcases:

        \begin{enumerate}[label=2.\arabic*]
        
            \item\emph{If $|A_{k+1}| \le n^{1 - 2\gamma}$:}
            In this case we can just run a standard 3SUM algorithm in time $O(|A_{k+1}|^2)$. By the choice of $k$, we have $|A_{k+1}| \le n^\gamma|A_k|$. So the time per eliminated element in $A_k$ is 
            \begin{equation*}
                O\parens*{\frac{|A_{k+1}|^2}{|A_k|}} = O\parens*{|A_{k+1}|n^\gamma} = O\parens*{n^{1-\gamma}}.
            \end{equation*}
            
            \item\emph{If $|A_{k+1}| > n^{1 - 2\gamma}$:}
            In this case, we turn $A_{k+1}$ into a modeled instance, remove it, and return it as one of our instances. More specifically, by the construction of $S_{k+1}$, we can construct a $(g,s)$-modeling of $A_{k+1}$ as follows. Define $\model:A_{k+1}\to S_{k+1}$ where $\model(a)$ is the starting point in $S_{k+1}$ of the interval in $\mathcal{I}_{k+1}$ that contains $a$. By definition, $\phi(A_{k+1})$ is $(\Delta, O(1))$-spaced, and by the choice of $\Delta$, any width-$\Delta$ interval covers at most $s$ elements. Further, $g = |S_{k+1}| \le 9n / t_{k+1} \le 9n^{1 - \delta/2}$. Also, for any triple $a,b,c \in A_{k+1}$ with $a + b + c = 0$, we have $\model(a) + \model(b) \in [-\model(c) - 2\Delta, -\model(c) + \Delta]$.
            Therefore, $(\phi, 2\Delta)$ is a $(g,s)$-modeling.
        
            Notice that whenever we reach this case, we remove at least $|A_k| \ge |A_{k+1}| / n^\gamma \ge n^{1 - 3 \gamma}$ elements. So in total we construct at most $n^{3\gamma}$ instances.
        \end{enumerate}

    \end{enumerate}
    
    \paragraph{Running Time} In total, there are at most $n^{1-\delta}$ iterations, as we eliminate at least $s$ elements per iteration. Constructing the sets $A_1, \dots, A_{\ceil{1/\gamma}}$ takes time $\tilde O(n/\gamma) = \tilde O(n)$. The total time spent removing elements in case 1 is at most $\tilde O(n^{2-\delta\gamma/2})$. In case 2.1, we spend at most $O(n^{2-\gamma})$ removing elements. In case 2.2, we construct a modeling in linear time. Overall, the $\tilde O(n^{2 - \delta\gamma/2})$ term dominates.
\end{proof}

In the second step, we show how to reduce a modeled instance to a Real Convolution 3SUM instance.
We invoke our machinery based on the Frank--Tardos theorem and additive combinatorics from \Cref{sec:toolkit}, to turn a subset of our model numbers $\phi(A)$ into integers. These integers form the indices for the Real Convolution instance.
All remaining model numbers contribute only few approximate solutions. We can list all of them in time $O(g^2)$ and check for true solutions within the buckets.

\begin{lemma} \label{lem:reduction-to-integer-modeled}
    Suppose that Real Convolution 3SUM is in time $O(n^{2 - \epsilon})$, for some constant \makebox{$\epsilon > 0$}.
    Then any $(g,s)$-modeled Real 3SUM instance $A$ can be solved in time $\tilde O(g^2 + s^3 g^{2-\Omega(\epsilon)})$.
\end{lemma}
\begin{proof}
    Let $A$ denote the given 3SUM instance, and let $(\model, \Delta)$ denote the $(g,s)$-modeling of $A$. Run the popular/unpopular decomposition (\Cref{lem:decomp-pop}) on $\model(A)$, with threshold $\Delta$ and some parameter~$K$ to be determined later, to obtain a partition $\model(A) = P \sqcup U$ satisfying
    \begin{equation*}
        E^{\pm \Delta}(\phi(A), U) < \frac{|\phi(A)|^3}{K} \leq \frac{g^3}{K}
    \end{equation*}
    and, for each $a \in P$,
    \begin{equation*}
        r^{\pm \Delta}_{P+P-P}(a) \ge \frac{|\phi(A)|^2}{4K} \geq \frac{|P|^2}{4K}.
    \end{equation*}
    We distinguish two types of solutions $a + b + c = 0$ with $a,b,c \in A$:

    \begin{enumerate}
        \item \emph{(Partially Unpopular Solutions)}
        We first argue how to detect any solution $a, b, c$ for which at least one of $\model(a), \model(b), \model(c)$ is part of $U$. To this end, we list \emph{all} approximate solutions $x + y + z = 0 \pm\Delta$ for $x, y, z \in \model(A)$, where at least one of $x, y, z$ is part of $U$. Since $\model(A)$ is $(\Delta, O(1))$-spaced, each choice of $(x,y)$ leaves only $O(1)$ choices of $z$, so we can enumerate all candidates in time $O(g^2)$. But as $E^{\pm \Delta}(\phi(A),U) < g^3/K$, \Cref{lem:low-energy-few-solutions} implies that 
        we actually only find $O(g^2 / \sqrt{K})$ approximate solutions. For each such solution $x + y + z = 0 \pm  \Delta$, we enumerate all $a \in \model^{-1}(x)$ and $b \in \model^{-1}(y)$, and check if $-(a+b) \in A$, i.e., if they form a solution in $A$. In total, we enumerate at most $O(g^2 / \sqrt{K} \cdot s^2)$ pairs $(a,b)$.
    
        \item \emph{(Fully Popular Solutions)}
        To find solutions with all elements in $P$, we use the algorithm for Real Convolution 3SUM. We first apply \Cref{lem:hash-popular-to-int} to the popular elements $P$ with the parameter $4K$, and we receive a hash function $h_1 : P \to [\pm 2^{(K \log g)^{O(1)}}]$.
        
        The remaining algorithm is a standard application of linear hashing to reduce to a Real Convolution 3SUM instance, which we solve using the supposed algorithm.
        Apply \Cref{lem:linear-hashing} to obtain a second hash function $h_2 : [\pm 2^{(K \log g)^{O(1)}}] \to [g]$ as well as a constant-sized set $\Lambda$.
        Finally, for $a \in A$ define $h(a) \coloneqq h_2(h_1(\model(a))) - \lambda_a$, where $\lambda_a$ is a uniformly random element from $\Lambda \cup \set{-h_2(0)}$.
        Hence, $h$ is a function $A \to [\pm 2g]$.
        
        Combining \Cref{def:modeled,lem:hash-popular-to-int,lem:linear-hashing}, we claim that the following properties are satisfied.
        \begin{enumerate}[label=(\roman*)]
            \item For all $a,b,c \in A$ with $a + b + c = 0$, we have $\Pr[h(a) + h(b) + h(c) = 0] \ge \Omega(1)$.
            \item For all $a \in A$, the expected bucket size is $O(s)$, i.e., $\mathbb{E}\left[\abs{\set{b \in A : h(b) = h(a)}}\right] \le O(s)$.
        \end{enumerate}

        Indeed, for $a,b,c \in A$ with $a+b+c=0$, by the definition of modeling (\Cref{def:modeled}), we have $\model(a) + \model(b) + \model(c)=0 \pm \Delta$. After applying the subset sum cover and Frank--Tardos (\Cref{lem:hash-popular-to-int}), this implies $\Pr[h_1(\model(a)) + h_1(\model(b)) + h_1(\model(c)) = 0] \ge \Omega(1)$. Conditioned on this event, the almost-additivity of hashing from \Cref{lem:linear-hashing} gives 
        \begin{align*}
            h_2(0) &= h_2\big(h_1(\model(a)) + h_1(\model(b)) + h_1(\model(c))\big)\\
            &= h_2\big(h_1(\model(a)) + h_1(\model(b))\big) + h_2\big(h_1(\model(c))\big) + \lambda\\
            &= h_2\big(h_1(\model(a))\big) + h_2\big(h_1(\model(b))\big) + h_2\big(h_1(\model(c))\big) + \lambda + \lambda',
        \end{align*}
        for some $\lambda, \lambda' \in \Lambda$.
        Hence, $h(a) + h(b) + h(c) = 0$ in the event that $\lambda_a = -h_2(0)$, $\lambda_b = \lambda$ and $\lambda_c = \lambda'$, so property (i) holds.
        
        Regarding property (ii), \Cref{def:modeled} gives $\abs{\set{b \in A : \model(b) = \model(a)}} \le s$. After \Cref{lem:hash-popular-to-int}, this becomes $\abs{\set{b \in A : h_1(\model(b)) = h_1(\model(a))}} \le O(s)$. By pairwise independence of the linear hashing (\Cref{lem:linear-hashing}) and the union bound, we get $\Pr[h_2(x) = h_2(y)] \le O(1/g)$ for any distinct~$x, y$. By linearity of expectation, the expected number of $h_1$-values per bucket is constant and property (ii) holds.
        
        We call an element $a \in A$ \emph{bad} if $\abs{\set{b \in A : h(b) = h(a)}} \ge st$, for a parameter $t$ to be determined later. We detect solutions including at least one bad element separately from those without bad elements.
        \begin{enumerate}
            \item \emph{(Solutions with a Bad Element)}
            By Markov's inequality, property (ii) implies that the expected number of bad elements is at most $O(gs / t)$. Another application of Markov's inequality yields that the number of bad elements is at most $O(gs/t)$ with constant probability.
            Conditioned on this event, by enumerating all bad elements $a \in A$ and all $b \in A$, we can detect solutions involving bad elements in time $O(g^2s^2 / t)$.
            
            \item \emph{(Solutions without Bad Elements)}
            For the remaining \emph{good} elements, create a sequence $\Real^{[\pm 2g]}$ that contains in the $i$-th position an element chosen uniformly at random from $h^{-1}(i)$ (in case $h^{-1}(i) = \emptyset$, just place a large enough real number such that it is guaranteed not to appear in solutions). Use the Real Convolution 3SUM algorithm to solve this instance in time $O(g^{2-\epsilon})$. Since buckets have size $O(st)$, a solution $a,b,c \in A$ is preserved with probability at least $\Omega(1 / (st)^3)$ by property (i).
            By creating $O(s^3t^3\log (gs))$ such instances, we detect a good solution with high probability.
        \end{enumerate}
    \end{enumerate}

    \paragraph{Running Time}
    Computing the partition $\model(A) = P \sqcup U$ with \Cref{lem:decomp-pop} takes time $\tilde O(gK^2)$. In case 1 we spend $O(g^2)$ time computing approximate solutions on the model numbers, plus $O(s^2g^2/K)$ time enumerating solutions in $A$. In total, step 1 takes time $O\parens*{g^2 + s^2g^2/K}$.

    In case 2, we first spend time $\tilde O(gK + (K\log g)^{O(1)})$ to compute $h$ with \Cref{lem:hash-popular-to-int,lem:linear-hashing}. Finding solutions involving bad elements takes time $O(g^2s^2/t)$ with constant probability, and repeating it $O(\log (gs))$ times gives the same bound with high probability. For solutions with only good elements, we invoke the Real Convolution 3SUM algorithm $O(s^3t^3\log (gs))$ times on an instance of size $g$, hence taking time $\tilde O(s^3t^3g^{2-\epsilon})$. In total, detecting solutions with only popular model numbers takes time $\tilde O\parens*{gK^{O(1)} + g^2s^2/t + s^3t^3g^{2-\epsilon}}$.
    
    In total, the algorithm runs in time
    \begin{equation*}
        \tilde O\parens*{gK^{O(1)} + g^2 + g^2s^2/\sqrt{K} +  g^2s^2/t + s^3t^3\parens*{gK^{O(1)} + g^{2-\epsilon}}}.
    \end{equation*}
    Choosing $t = g^\beta$ and $K = g^{2\beta}$ for some constant $\beta > 0$, this becomes to
    \begin{equation*}
        \tilde O\parens*{g^2 + s^2 g^{2-\beta} + s^3 g^{2-\epsilon+O(\beta)} + s^3 g^{1+O(\beta)}}.
    \end{equation*}
    This is as claimed, $\tilde O\parens*{g^2 + s^3g^{2-\Omega(\epsilon)}}$, when choosing $\beta = \Theta(\epsilon)$ sufficiently small so that the $O(\beta)$ terms are at most $\epsilon/2$.
\end{proof}

\begin{proof}[Proof of \Cref{thm:real-to-real-conv}]
    First, call \Cref{lem:reduction-to-clustered} with parameters $\gamma, \delta$ to be determined later to obtain up to $n^{3\gamma}$ instances that are $(g,s)$-modeled for $g= 9n^{1-\delta/2}$ and $s= n^\delta$ in time $\tilde O(n^{2-\delta\gamma/2})$.
    Next, use \Cref{lem:reduction-to-integer-modeled} to solve each of them with the algorithm for Real Convolution 3SUM in time $\tilde O(g^2 + s^3g^{2-\Omega(\epsilon)})$.
    The total running time is
    \begin{equation*}
        \tilde O\parens*{n^{2 - \delta\gamma/2} + n^{3\gamma}\parens*{g^2 + s^3g^{2-\Omega(\epsilon)}}} = \tilde O\parens*{n^{2 - \delta\gamma/2} + n^{2 + 3\gamma-\delta} + n^{2 + 3\gamma + 2\delta - \Omega(\epsilon)}}.
    \end{equation*}
    Choosing $\delta = 4\gamma$ turns the expression into 
    \begin{equation*}
        \tilde O\parens*{n^{2-2\gamma^2} + n^{2-\gamma} + n^{2+11\gamma - \Omega(\epsilon)}}.
    \end{equation*}
    This becomes $n^{2-\Omega(\epsilon^2)}$ by setting $\gamma = \Theta(\epsilon)$ to a sufficiently small constant so that the $\Omega(\epsilon)$ term becomes at least $30\gamma$.
\end{proof}

%% file: sections/real-conv-to-int.tex
\section{Reduction from Real Convolution 3SUM to Integer 3SUM} \label{sec:real-conv-to-int}

In this section we give our reduction from Real Convolution 3SUM to Integer 3SUM. Together with \Cref{thm:real-to-real-conv}, this completes the proof of \Cref{thm:main}.

\begin{theorem}\label{thm:conv-to-integer-reduction}
If Integer 3SUM can be solved in time $O(n^{2-\epsilon})$, for some constant $\epsilon > 0$, then Real Convolution 3SUM can be solved in time \smash{$O(n^{2-\Omega(\epsilon^2)})$}.
\end{theorem}

On a high level, the proof proceeds by first rounding all numbers to a suitable scale, i.e., computing $\Delta$ such that we have many solutions up to $\pm \Delta$, but not too many such that we can still list them.
Then, in \Cref{sec:conv-decomp}, we split again into \emph{popular} and \emph{unpopular} elements, that is, elements participating in many and few solutions energy quadruples, and hence solutions due to \Cref{lem:low-energy-few-solutions} (for technical reasons, we also work with an intermediate, \emph{quasi-popular} level).
In \Cref{sec:type12}, we explain how to list approximate solutions where every element is at least quasi-popular using our machinery from \Cref{sec:toolkit}.
To list solutions with both unpopular and popular elements, in \Cref{sec:type3} we develop a solution akin to Fredman's trick~\cite{Fredman76,ChanVX23}.
Finally, in \Cref{sec:type4}, we can remove all popular elements from the instance, and repeat the whole algorithm until every element was popular for at least one $\Delta$ or the number of triples left is small enough to list all of them.

In \Cref{sec:listing-to-decision,sec:indexed-sets}, we state two more necessary tools for the proof: a listing-to-decision reduction for 3SUM, and an equivalent version of Real Convolution 3SUM that we work with.

\subsection{Listing-to-Decision Reduction for 3SUM}\label{sec:listing-to-decision}

In this section we argue that if (the decision variant of) Integer 3SUM can be solved in subquadratic time, then it is also possible to \emph{list} any subquadratic number of solutions in subquadratic time. This proof mostly follows standard arguments~\cite{VassilevskaW18,DurajKPW20}, but to get exactly the bound we need, we cannot reuse them in a black-box manner, and we need to prove the following lemma.

\begin{lemma}[Listing-to-Decision Reduction for 3SUM] \label{lem:3sum-listing-to-decision}
Assume that there is a decision algorithm for Integer 3SUM running in time $O(n^{2-\varepsilon})$, for some $\varepsilon > 0$. Then there is also an algorithm that, given three size-$n$ sets of integers $A, B, C$, and a number $T$, lists up to $T$ solution triples $(a,b,c) \in A \times B \times C$ with $a+b+c=0$ in time $O(n^{2-\varepsilon}T^{\varepsilon/2} \log n)$. 
\end{lemma}

\begin{proof}
For each of the three input sets, sort it and split it into $O(\sqrt{T})$ groups of consecutive elements, each group of size $O(n/\sqrt{T})$. There are $O((\sqrt{T})^3)$ triples of groups (one group per set), which we are going to call \emph{subinstances}. Each solution triple lives in exactly one subinstance. However, by a standard argument~\cite{LincolnWWW16,FischerKP24}, there are only $O((\sqrt{T})^2)=O(T)$ nontrivial subinstances that may contain solutions, and they can be efficiently identified by examining the minimum and maximum elements in each group, and requiring that the sum of the three minimums is at most $0$ and the sum of the three maximums is at least $0$.

For $\lceil\log(n/\sqrt{T})\rceil$ rounds, repeat the following: Split the first set in each subinstance into two equisized parts, effectively doubling the number of subinstances. Run the decision algorithm on each subinstance\footnote{In order to run a decision algorithm for a one-set variant of 3SUM on three sets one can use the standard reduction between the two variants, see, e.g., \cite{ChanH20}.} and if it says ``no'', discard this subinstance. At the end of the round, if the number of subinstances exceeds $T$, discard arbitrarily chosen ones so that their count drops back to~$T$.

After the last round, the first set in each subinstance is just a singleton, so all solutions in this subinstance can be listed by a linear scan in time $O(n/\sqrt{T})$, which is no more than the algorithm anyway spends in a single round.

The correctness follows from the invariant that solutions coming from different subinstances are distinct, and that the total number of solutions across non-discarded subinstances is at least $T$, or the total number of solutions was initially less than $T$ in which case the algorithm does not discard anything.

The total running time of a single round is $O(T (n/\sqrt{T})^{2-\epsilon}) = O(n^{2-\varepsilon}T^{\varepsilon/2})$, and there are $O(\log n)$ rounds, which gives the claimed total running time.
\end{proof}

\subsection{Convolutions and Indexed Sets}\label{sec:indexed-sets}

In this section, it is more convenient to work with input sequences indexed by integers in $[\pm n] \coloneqq \{-n,\ldots,n\}$ and allowing missing entries. This leads to the following version of Real Convolution 3SUM: The input consists of a sequence $a \in (\Real \cup \set{\bot})^{[\pm n]}$. A \emph{solution} is a triple $(i, j, k) \in [\pm n]^3$ such that $i + j + k = 0$ and $a[i] + a[j] + a[k] = 0$ (where none of these three equals $\bot$). The dummy  entry $\bot$ symbolizes a missing entry, and any expression containing $\bot$ can never satisfy an equation or form a solution. We call triple $(i, j, k) \in [\pm n]^3$ a \emph{$\Delta$-approximate solution} if $i + j + k = 0$ and $a[i] + a[j] + a[k] = 0 \pm \Delta$.

Moreover, it is often convenient to transition from sequences to sets inheriting the convolution structure. Specifically, many of the additive combinatorics statements that we have previously proved for sets carry over easily to sequences, but it would be tedious to reprove each single one. Instead, we apply these results as black-boxes, relying on the following notation. For a sequence $a \in (\Real \cup \set{\bot})^{[\pm n]}$, define the map $\iota_a : [\pm n] \to \Real$ defined by $\iota_a(i) \coloneqq iW + a[i]$, and define the \emph{indexed set}
\begin{align*}
    I_a &\coloneqq \set{\iota_a(i) : i \in [\pm n],\, a[i] \neq \bot}.
\end{align*}
Here, $W$ is an arbitrary sufficiently large number. The choice of $W$ does not really matter, as long as it is larger than the entries in $a$, and larger than the thresholds $\Delta$ considered throughout (by a constant factor, say). The point is that the additive structure of $a$ and $I_a$ is identical. The number of $\Delta$-approximate solutions in the sequence version is exactly the number of solutions $a_1 + a_2 + a_3 = 0 \pm \Delta$ for $a_1, a_2, a_3 \in I_a$.  For example, also the number of ``additive quadruples''
\begin{equation*}
    \begin{cases}
        i + j = k + \ell,\\
        a[i] + a[j] = a[k] + a[\ell] \pm \Delta,
    \end{cases}
\end{equation*}
is exactly
\begin{equation*}
    |\set{(a_1, a_2, a_3, a_4) \in I_a^4 : a_1 + a_2 = a_3 + a_4 \pm \Delta}| = E^{\pm \Delta}(I_a).
\end{equation*}

\subsection{Decomposition into Popular, Quasi-Popular, and Unpopular Elements}\label{sec:conv-decomp}

The following lemma chooses a threshold $\Delta$ such that the number of $2\Delta$-approximate solutions is large, whereas the number of $\Delta$-approximate solutions is not too large. The idea is to imagine a sorted order of all $|a[i] + a[j] + a[k]|$ where $i+j+k=0$. Choosing $2\Delta$ as a sum of roughly the correct rank satisfies the properties.

\begin{lemma}[Approximation Threshold]\label{lem:approx-threshold}
    Let $a \in (\Real \cup \set{\bot})^{[\pm n]}$ be a sequence and let $K \geq 1$. There is a randomized algorithm that runs in time $K (\log n)^{O(1)}$ and either detects a solution, decides that the number of $\infty$-approximate solutions is at most $O(n^2 / K)$, or computes some $\Delta \in \Realnneg$ such that:
    \begin{enumerate}[label=(\roman*)]
        \item the number of $\Delta$-approximate solutions is at most $O(n^2 / K)$, and
        \item the number of $2\Delta$-approximate solutions is at least $\Omega(n^2 / K)$.
    \end{enumerate}
\end{lemma}
\begin{proof}
    Let $N \coloneqq \Theta(n^2)$ denote the number of triples $i, j, k \in [\pm n]$ with $i + j + k = 0$. Imagine the length-$N$ sequence of numbers $|a[i] + a[j] + a[k]|$, for all triples $i, j, k$ with $i + j + k = 0$, sorted in increasing order (where it is understood that all triples with a $\bot$ entry take the largest positions in the order). Our goal is to find one such triple $(i^*, j^*, k^*)$ whose rank in the sorted order is at least $N / K$ and at most $2N / K$. There are two special cases: If $a[i^*] + a[j^*] + a[k^*] = 0$, then we have directly found a solution. If one of the three entries $a[i^*], a[j^*], a[k^*]$ is $\bot$, then the number of non-$\bot$ triples and hence the number of $\infty$-approximate solutions is at most $2n^2 / K$. Otherwise, we assign $2\Delta = |a[i^*] + a[j^*] + a[k^*]|$. Since the rank is at least $N / K$, there are at least $N / K$ other triples $(i, j, k)$ with
    \begin{equation*}
        |a[i] + a[j] + a[k]| \leq |a[i^*] + a[j^*] + a[k^*]| = 2\Delta.
    \end{equation*}
    In other words, the number of $2\Delta$-approximate solutions is at least $N / K$. Similarly, there are at most $2N / K$ triples $(i, j, k)$ with
    \begin{equation*}
        |a[i] + a[j] + a[k]| \leq \Delta < |a[i^*] + a[j^*] + a[k^*]|,
    \end{equation*}
    using that $\Delta > 0$. In particular, the number of $\Delta$-approximate solutions is at most $2n^2 / K$.
    
    To find $(i^*, j^*, k^*)$ we apply a simple sketching technique. Let $R \coloneqq \Theta(\log n)$ be a parameter. Sample $R K$ triples $(i, j, k)$ with $i + j + k = 0$ uniformly at random, and pick the one for which $|a[i] + a[j] + a[k]|$ is $1.5R$-th smallest among the sampled triples. Equivalently, we choose $R K$ uniformly random positions in the sorted order, and commit to the $1.5R$-th smallest position among the samples. Each sample has rank below $N / K$ with probability exactly $1/K$. Hence, the expected number of samples below rank $N / K$ is $R K / K = R$. Therefore, by Chernoff's bound it follows that with high probability there are less than $2R$ samples below rank $n^2 / K$. In particular, the lower bound on the $2\Delta$-approximate-solutions follows. Similarly, the expected number of samples below rank $2N / K$ is $2R$, and by Chernoff's bound with high probability there are at least $1.5R$ such samples. Thus, the upper bound on the $\Delta$-approximate-solutions, property (ii), holds too.
\end{proof}

Next, we state the decomposition into popular, quasi-popular, and unpopular elements for the convolution setting. The decomposition immediately follows from \Cref{lem:decomp-pop}.

\begin{lemma}[Popular/Unpopular Decomposition] \label{lem:conv-decomp-pop}
    Let $a \in (\Real \cup \set{\bot})^{[\pm n]}$ be a sequence, let $\Delta \in \Realnneg$, and let $K \geq 1$ be an integer parameter. There is a partition $[\pm n] = P \sqcup U$ into \emph{popular} and \emph{unpopular} indices such that:
    \begin{enumerate}[label=(\roman*)]
        \item For all $i \in P$: $\abs{\set{(j, k, \ell) \in P^3 : i + j = k + \ell, a[i] + a[j] = a[k] + a[\ell] \pm \Delta}} \geq |P|^2 / 4K$,
        \item $\abs{\set{(i, j, k, \ell) \in U^2 \times [\pm n]^2 : i + j = k + \ell, a[i] + a[j] = a[k] + a[\ell] \pm \Delta}} \leq n^3 / K$.
    \end{enumerate}
    It can be computed in time $\tilde O(nK^2)$.
\end{lemma}
\begin{proof}
We apply \Cref{lem:decomp-pop} on the indexed set $A = I_a$, which yields a partition $A = \hat P \sqcup \hat U$. Then let $P \subseteq [\pm n]$ be the subset so that $\iota_a(P) = \hat P$ and let $U \coloneqq [\pm n] \setminus P$. \Cref{lem:decomp-pop} states that for each $a \in A$, \smash{$r^{\pm \Delta}_{\hat P+\hat P-\hat P}(a) \geq |I_a|^2 / 4K \geq |P|^2 / 4K$}. Pulling back through $\iota_a$, this implies~(i). \Cref{lem:decomp-pop} further implies that $E^{\pm}(I_a, \hat U) < |I_a| / K^3 \leq n / K^3$. Pulling back through $\iota_a$, this implies~(ii). The running time is $\tilde O(n K^2)$ as in \Cref{lem:decomp-pop}.
\end{proof}

\begin{corollary}[Popular/Quasi-Popular/Unpopular Decomposition] \label{lem:conv-decomp-quasi-pop}
    Let $a \in (\Real \cup \set{\bot})^{[\pm n]}$ be a sequence, let $\Delta \in \Realnneg$, and let $1 \leq L \leq M$ be integer parameters. There is a partition $[\pm n] = P \sqcup Q \sqcup U$ into \emph{popular}, \emph{quasi-popular}, and \emph{unpopular} indices such that:
    \begin{enumerate}[label=(\roman*)]
        \item For all $i \in P$: $\abs{\set{(j, k, \ell) \in P^3 : i + j = k + \ell, a[i] + a[j] = a[k] + a[\ell] \pm 2\Delta}} \geq |P|^2 / 4L$,
        \item For all $i \in Q$: $\abs{\set{(j, k, \ell) \in (P \cup Q)^3 : i + j = k + \ell, a[i] + a[j] = a[k] + a[\ell] \pm 2\Delta}} \geq |P \cup Q|^2 / 4M$,
        \item $\abs{\set{(i, j, k, \ell) \in Q^2 \times [\pm n]^2 : i + j = k + \ell, a[i] + a[j] = a[k] + a[\ell] \pm 2\Delta}} \leq n^3 / L$, and
        \item $\abs{\set{(i, j, k, \ell) \in U^2 \times [\pm n]^2 : i + j = k + \ell, a[i] + a[j] = a[k] + a[\ell] \pm 2\Delta}} \leq n^3 / M$.
    \end{enumerate}
    It can be computed in time $\tilde O(nM^2)$.
\end{corollary}
\begin{proof}
    The proof is by two applications of the previous \cref{lem:conv-decomp-pop}: First, we apply the lemma on the sequence $a$, with threshold $2\Delta$ and parameter $M$, to obtain a partition $[\pm n] = S \sqcup U$. Then we apply the lemma again on the sequence $a$ restricted to the indices in $S$, with threshold $2\Delta$ and with parameter $L$. This yields a partition $S = P \sqcup Q$. The four properties are immediate.
\end{proof}

For the remaining reduction, let us separate solutions into cases with respect to where they lie in the decomposition $P \sqcup Q \sqcup U$.
We claim that every solution falls into one of the following four types:
\begin{enumerate}
    \item only popular elements, i.e., $P^3$,
    \item at least one popular, one quasi-popular, and no unpopular elements, i.e., $P \times Q \times (P \cup Q)$,
    \item at least one popular and at least one unpopular element, i.e., $P \times U \times [\pm n]$, and
    \item no popular elements, i.e., $(Q \cup U)^3$.
\end{enumerate}
Every solution without elements from $P$ falls into Type 4, so we only need to argue that the first three types cover all solutions including popular elements. Type 3 covers every remaining solution involving at least one unpopular element. Out of the remaining solutions with at least one element in $P$ and none in $U$, Type 2 covers every solution with at least one element in $Q$. Every other solution is left with only elements from $P$, hence it falls into Type 1.
So, these four types indeed cover every possible Real Convolution 3SUM solution.
In the remaining section we show how to detect an exact solution of each type separately.

\begin{lemma}[Number of Approximate Solutions] \label{lem:approx-sol-count}
Assume that $L$ is larger than $K^2$ by a sufficiently large constant factor. Then:
\begin{itemize}
    \item the number of $2\Delta$-approximate solutions of Type 1 is at least $\Omega(n^2 / K)$,
    \item the number of $2\Delta$-approximate solutions of Types 2, 3, and 4 is at most $O(n^2 / L^{1/2})$.
\end{itemize}
\end{lemma}
\begin{proof}
Let $\hat A \coloneqq I_a$, let $\hat Q \coloneqq \set{\iota_a(i) : i \in P}$ and $\hat U \coloneqq \set{\iota_a(i) : i \in P}$, and define $\hat T \coloneqq \hat Q \cup \hat U$. Each $2\Delta$-approximate solutions of Type 2 or 3 involves at least one quasi-popular item. We can thus upper bound these solutions by the number of 3SUM solutions $x + y + z = 0 \pm 2\Delta$ for $x, y, z \in \hat A$, where at least one item is part of $\hat Q$. By \Cref{lem:low-energy-few-solutions}, the number of such solutions is at most
\begin{equation*}
    O\parens*{\sqrt{E^{\pm 2\Delta}(\hat A, \hat Q) |\hat A|}}.
\end{equation*}
\cref{lem:conv-decomp-quasi-pop}~(iii) gives an upper bound of $E^{\pm 2\Delta}(\hat A, \hat Q) \leq n^3 / L$, leading to the stated bound. The number of $2\Delta$-approximate Type-4 solutions can be bounded analogously with \cref{lem:conv-decomp-quasi-pop}~(iv) (and recalling that $L \leq M$).

In summary, the number of $2\Delta$-approximate solutions of Types 2, 3 and 4 is $O(n^2 / L^{1/2})$. However, \cref{lem:approx-threshold} states that the total number of $2\Delta$-approximate solutions is at least $\Omega(n^2 / K)$. Thus, when $L$ and $M$ are larger than $K^2$ by a sufficiently large factor, the number of $2\Delta$-approximate solutions of Type 1 must be at least $\Omega(n^2 / K)$.
\end{proof}

\subsection{Types 1 and 2: Reduction to 3SUM Listing}\label{sec:type12}

We find solutions of Types 1 and 2 in a similar way. For both solution types, the elements are popular enough to be hashed to integers with our machinery from \Cref{sec:toolkit}. After reducing $\Delta$-approximate solutions to exact integer solutions, we need to list all of them. By the choice of $\Delta$ in \Cref{lem:approx-threshold}, there are not too many solutions in total, so we can afford to list all of them for Type~1.

\begin{lemma}[Type-1 Solutions]\label{lem:type1}
    Assume that Integer 3SUM can be solved in time $O(n^{2-\epsilon})$, for some constant $\epsilon > 0$. Then there is a randomized algorithm to detect any solution in $P^3$ in time
    \begin{equation*}
        \tilde O\parens*{nL^{O(1)} + \frac{|P| n}{K^{\epsilon / 2}}}.
    \end{equation*}
\end{lemma}
\begin{proof}
    Consider the indexed set $\hat P = \set{\iota_a(i) : i \in P}$, and observe that this set is $\Delta/100$-spaced (by choosing $W$ large enough). \Cref{lem:conv-decomp-quasi-pop} implies that
    \begin{equation*}
        r^{\pm 2\Delta}_{\hat P + \hat P - \hat P}(x) \ge \frac{|P|^2}{4L},
    \end{equation*}
    for all $x \in \hat P$.
    Therefore, we can apply our subset sum cover and Frank--Tardos machinery from \Cref{lem:hash-popular-to-int} to the popular elements $\hat P$ with threshold $\Delta / 100$ and constant $C = 200$. This yields a hash function $h_1 : \hat P \to [\pm 2^{(L \log n)^{O(1)}}]$, and runs in time $\tilde O(nL + (L\log n)^{O(1)})$.
    
    The remaining proof first applies a standard universe reduction for the resulting integer instance using linear hashing as in \Cref{lem:linear-hashing}. Finally, we show that all solutions can be listed efficiently.
    Invoke \Cref{lem:linear-hashing} to get a second hash function $h_2 : [\pm 2^{(L \log n)^{O(1)}}] \to [m]$ for $m = n^{O(1)}$ and define $h(x) \coloneqq h_2(h_1(x)) - \lambda_x$ for $x \in \hat P$ and a value $\lambda_x$ chosen uniformly at random from $\Lambda_0 \coloneqq \Lambda \cup \set{-h_2(0)}$. Combining \Cref{lem:hash-popular-to-int,lem:linear-hashing}, we claim that the following two properties hold.
    \begin{enumerate}[label=(\roman*)]
        \item For all $x,y,z \in \hat P$ with $x + y + z = 0 \pm \Delta/100$, we have $\Pr[h(x) + h(y) + h(z) = 0] \ge \Omega(1)$.
        \item With high probability, there are no \emph{false positives} $(x, y, z) \in \hat{P}^3$ with $h(x) + h(y) + h(z) = 0$ and $x + y + z \ne 0 \pm \Delta$. 
    \end{enumerate}
    
    Indeed, for $x,y,z \in \hat P$ with $x + y + z = 0 \pm \Delta / 100$ we have $\Pr[h_1(x) + h_1(y) + h_1(z) = 0] \ge \Omega(1)$ by \Cref{lem:hash-popular-to-int}. Conditioned on this event, almost-additivity of $h_2$ implies
    \begin{align*}
        h_2(0) &= h_2\big(h_1(x) + h_1(y) + h_1(z)\big)\\
        &= h_2\big(h_1(x) + h_1(y)\big) + h_2\big(h_1(z)\big) + \lambda\\
        &= h_2\big(h_1(x)\big) + h_2\big(h_1(y)\big) + h_2\big(h_1(z)\big) + \lambda + \lambda',
    \end{align*}
    for some $\lambda, \lambda' \in \Lambda$.
    Hence, $h(x) + h(y) + h(z) = 0$ in the event that $\lambda_x = -h_2(0)$, $\lambda_y = \lambda$ and $\lambda_z = \lambda'$, so property (i) holds.
    
    Regarding property (ii), by \Cref{lem:hash-popular-to-int}, for all $x,y,z\in \hat P$ with $h_1(x) + h_1(y) + h_1(z) = 0$, we have $x + y + z = 0 \pm \Delta$.
    So fix any triple $x,y,z \in \hat P$ with $h_1(x) + h_1(y) + h_1(z) = 0$, and fix any bucket~$k \in [m]$. The probability that
    \begin{equation*}
        h_2\big(h_1(x) + h_1(y)\big) = h_2\big(h_1(z)\big) + k
    \end{equation*}
    is at most $O(1/m)$ by pairwise independence (\Cref{lem:linear-hashing}). Moreover, the probability of this event for any choice $k \in \Lambda_0 + \Lambda_0 + \Lambda_0 - \Lambda_0$ is at most $O(1/m)$ by a union bound over $O(1)$ options. Conditioned on this event not happening, $(x, y, z)$ is not a false positive, as otherwise the almost-linearity would imply that
    \begin{align*}
        0 &= h(x) + h(y) + h(z) \\
        &= h_2\big(h_1(x)\big) + h_2\big(h_1(y)\big) + h_2\big(h_1(z)\big) - \lambda_x - \lambda_y - \lambda_z \\
        &= h_2\big(h_1(x) + h_1(y)\big) - h_2\big(h_1(z)\big) + \lambda - \lambda_x - \lambda_y - \lambda_z,
    \end{align*}
    for some $\lambda, \lambda_x, \lambda_y, \lambda_z \in \Lambda_0$. This shows that any $(x, y, z)$ becomes a false positive with probability at most $O(1/m)$. By a union bound over the $n^3$ triples, the probability of introducing any false positives is at most $O(n^3 / m) \leq 1/n^{\Omega(1)}$ for $m \ge n^4$, proving property~(ii).

    Putting both properties together, (i) with constant probability any fixed proper solution is retained in $h(\hat P)$, (ii) and at the same time we do not introduce any false positives. Thus, by our choice of threshold $\Delta$ using \Cref{lem:approx-threshold}, the number of solutions in $h(\hat{P})^3$ is at most $T = O(n^2 / K)$.
    By the listing-to-decision reduction (\Cref{lem:3sum-listing-to-decision}), we can list all of them in time \begin{equation*}
        \tilde O\parens*{n^{2-\epsilon} T^{\epsilon/2}} \le \tilde O\parens*{|P|^{2-\epsilon} \frac{n^\epsilon}{K^{\epsilon / 2}}} \le \tilde O\parens*{\frac{|P|n}{K^{\epsilon / 2}}}. 
    \end{equation*}
    For each solution listed in this way, we explicitly check whether its preimage under $h$ forms a solution.
    We repeat this algorithm $O(\log n)$ times to amplify the success probability.
\end{proof}

The algorithm for solutions of Type 2 follows the same structure as the Type 1 algorithm. The main difference is in the final argument to bound the number of listed solutions. In the previous lemma, we used the bound on the total number of $\Delta$-approximate solutions from \Cref{lem:approx-threshold}. This time, we get an even stronger bound on the number of solutions since one of the three elements is only quasi-popular and hence not involved in as many solutions as popular elements. Therefore, we apply the energy bound from \Cref{lem:conv-decomp-quasi-pop} that directly bounds the number of solutions via \Cref{lem:low-energy-few-solutions}.

\begin{lemma}[Type-2 Solutions]\label{lem:type2}
Assume that Integer 3SUM can be solved in time $O(n^{2-\epsilon})$, for some constant $\epsilon > 0$. Then there is a randomized algorithm to detect any solution in $P \times Q \times (P \cup Q)$ in time
\begin{equation*}
    \tilde O\parens*{nM^{O(1)} + \frac{n^2}{L^{\epsilon / 4}}}.
\end{equation*}
\end{lemma}
\begin{proof}
Consider the indexed sets $\hat P = \set{\iota_a(i) : i \in P}$ and $\hat Q = \set{\iota_a(i) : i \in Q}$, and write $\hat S = \hat P \cup \hat Q$. By properties (i) and (ii) of \Cref{lem:conv-decomp-quasi-pop}, we have
\begin{equation*}
    r^{\pm 2\Delta}_{\hat S+ \hat S- \hat S}(x) \ge \frac{|\hat S|^2}{4M}
\end{equation*}
for all $x \in \hat S$. Hence, we can apply \Cref{lem:hash-popular-to-int} on $\hat S$ with parameter $\Delta \gets \Delta / 100$ and constant~\makebox{$C = 200$} to obtain a hash function $h_1 : \hat S \to [\pm 2^{(M \log n)^{O(1)}}]$ in time $\tilde O(nM + (M\log n)^{O(1)})$.
Next, we apply \Cref{lem:linear-hashing} to get a second hash function $h_2 : [\pm 2^{(M \log n)^{O(1)}}] \to [m]$ for $m = n^{O(1)}$ as well as a constant-sized set $\Lambda$. For $x \in \hat P$, we define $h(x) \coloneqq h_2(h_1(x)) - \lambda_x$ for a uniform random choice of~$\lambda_x$ from $\Lambda \cup \set{-h_2(0)}$. Just as in \cref{lem:type1}, we have the following two properties:
\begin{enumerate}[label=(\roman*)]
    \item For all $x,y,z \in \hat S$ with $x + y + z = 0 \pm \Delta/100$, we have $\Pr[h(x) + h(y) + h(z) = 0] \ge \Omega(1)$.
    \item With high probability, there are no \emph{false positives} $(x, y, z) \in \hat{S}^3$ with $h(x) + h(y) + h(z) = 0$ and $x + y + z \ne 0 \pm \Delta$. 
\end{enumerate}

Only now starts the proof to differ from \cref{lem:type1}.
To bound the number of $\Delta$-approximate solutions in $\hat{S}^3$ with at least one element from $\hat Q$, notice that by property (iii) of the decomposition from \Cref{lem:conv-decomp-quasi-pop}, we have
\begin{equation*}
    E^{\pm 2\Delta}(\hat Q, \hat S) \le E^{\pm 2\Delta}(\hat Q, \hat A) \le O\parens*{\frac{n^3}{L}}.
\end{equation*}
In this case, \Cref{lem:low-energy-few-solutions} implies that the number of solutions $x + y + z = 0 \pm 2\Delta$, where $x, y, z \in \hat P \cup \hat Q$ and where at least one item is from $\hat Q$, is at most $T = O(n^2 / L^{1/2})$. Moreover, property~(ii) guarantees that these are the only solutions in $h(\hat P \cup \hat Q)$, where at least one item is from $h(\hat Q)$. We apply the listing-to-decision reduction (\Cref{lem:3sum-listing-to-decision}) to list all such solutions (in the tripartite instance $h(\hat Q) \times h(\hat P \cup \hat Q)^2$) in time $\tilde O(n^2 / T^{\epsilon/2}) = \tilde O(n^2 / L^{\epsilon / 4})$. By property~(i), any fix proper solution survives in the hashed instance and is thus correctly listed. Finally, to amplify the success probability we repeat the overall algorithm $O(\log n)$ times.
\end{proof}

\subsection{Type 3: Fredman's Trick for Convolutions}\label{sec:type3}
Recall that Type-3 solutions are those which involve at least one popular and at least one unpopular item. We cannot hope to replace the unpopular item by an integer as for Types 1 and 2, and for this reason we resort to a different strategy, loosely inspired by~\cite{ChanVX23}.

On the one hand, we exploit that Type-3 solutions involve an unpopular item to bound the total number of Type-3 solutions (as in \cref{lem:approx-sol-count}); we can therefore afford to list all such solutions. On the other hand, we exploit that one of the items is popular to implement the listing efficiently. Intuitively, the popular items should exhibit a lot of structure by \cref{lem:conv-decomp-quasi-pop}~(i). A good example is to think of the popular items as a linear function, $a[i] = i \cdot \alpha$. But this makes the popular items quite predictable, and we may hope to test whether $a[i] + a[j] + a[k] = 0 \pm \Delta$ without actually evaluating $a[i]$. In the example, it would be sufficient to test whether $\alpha \cdot i + a[j] + a[k] = 0 \pm \Delta$, which can be equivalently expressed as $a'[j] + a'[k] = 0 \pm \Delta$ for the rewritten sequence $a'[j] = a[j] + j \cdot \alpha$. But this remaining condition $a'[j] + a'[k] = 0 \pm \Delta$ involves only two independent indices $j$ and $k$, and we can therefore list all solutions in output-sensitive linear time. Turning this intuition into a formal argument requires some more effort.

\begin{lemma}\label{lem:few-solutions-unpop}
    The number of solutions $(j,k,p,q,r) \in U \times [\pm n]^4$ to 
    \begin{equation*} 
        \begin{cases}
            j + k + p + q = r, \\
            a[j] + a[k] + a[p] + a[q] = a[r] \pm 2\Delta
        \end{cases}
    \end{equation*}
    is bounded by $O(n^4 / M^{1/2})$.
\end{lemma}
\begin{proof}
    Let $\hat A = I_a$ denote the indexed set of $a$, and let $\hat U \subseteq \hat A$ denote the subset of elements corresponding to the entries in $U$. Then the number of solutions from the lemma statement is exactly
    \begin{equation*}
        \abs{\set{(u,b,c,d,e) \in \hat U \times \hat{A}^4 : u + b + c + d = e \pm 2\Delta}}.
    \end{equation*}
    This can be bounded by
    \begin{align*}
        &\abs{\set{(u,b,c,d,e) \in \hat U \times \hat{A}^4 : u + b + c + d = e \pm 2\Delta}}\\
        &\qquad=\sum_{d \in \hat A} \sum_{e \in \hat A} \abs{\set{(u,b,c) \in \hat U \times \hat A \times (\hat A + d - e): u + b + c = 0 \pm 2\Delta}}\\
        &\qquad\le O\parens*{\sum_{d \in \hat A} \sum_{e \in \hat A} \sqrt{E^{\pm 2\Delta}(\hat U, \hat A) \abs{\hat{A}}}}\\
        &\qquad\le O\parens*{\sum_{d \in \hat A} \sum_{e \in \hat A} \frac{n^2}{M^{1/2}}}\\
        &\qquad= O\parens*{\frac{n^4}{M^{1/2}}},
    \end{align*}
    where the first inequality is by \Cref{lem:low-energy-few-solutions} and the second inequality is by \Cref{lem:conv-decomp-quasi-pop}~(iv).
\end{proof}

\begin{lemma}[Type-3 Solutions]\label{lem:type3}
There is a randomized algorithm to detect any solution in $P \times U \times [\pm n]$ in time
\begin{equation*}
    \tilde O\parens*{n L + \frac{n^2 L}{M^{1/2}}}.
\end{equation*}
\end{lemma}
\begin{proof}
The algorithm runs in $\Theta(L \log n)$ iterations. In each iteration, we sample $g, h \in [\pm 2n]$ uniformly at random. If $g + h \not\in [\pm n]$ (which happens with probability at most $1 - \Omega(1)$), we simply skip this iteration. Otherwise, construct the sequences $b, c \in (\Real \cup \set{\bot})^{[\pm n]}$ defined by
\begin{align*}
    b[j] &\coloneqq
    \begin{cases}
        a[j] + a[g - j] &\text{if $g - j \in [\pm n]$,} \\
        \bot &\text{otherwise,}
    \end{cases} \\
    c[k] &\coloneqq
    \begin{cases}
        a[k] + a[h - k] &\text{if $h - k \in [\pm n]$,} \\
        \bot &\text{otherwise.}
    \end{cases}
\end{align*}
We then enumerate all pairs $(j, k) \in U \times [\pm n]$ with
\begin{equation} \label{eq:fredman}
    b[j] + c[k] = a[g + h] \pm 2\Delta.
\end{equation}
More specifically, we first sort $\set{c[k] : k \in [\pm n]}$. Then we enumerate all $j \in U$, and all~\makebox{$k \in [\pm n]$} satisfying \eqref{eq:fredman} by scanning over the respective width-$4\Delta$ interval. For each such pair $j, k$ we test if it forms a solution, $a[-j-k] + a[j] + a[k] = 0$. If we ever find a solution in this way, we stop and report it.

\paragraph{Correctness}
Suppose $(i, j, k) \in P \times U \times [\pm n]$ is a solution. We argue that the algorithm detects $(i, j, k)$ with high probability over all iterations. Call an iteration \emph{successful} if the sampled indices $g$ and $h$ satisfy $g - j, h - k \in [\pm n]$ and
\begin{equation} \label{eq:fredman-energy}
    a[i] + a[g + h] = a[g - j] + a[h - k] \pm 2\Delta.
\end{equation}
We argue that each iteration is successful with probability at least $\Omega(1/L)$. Indeed, with constant probability we satisfy $g - j, h - k \in [\pm n]$, and whenever this happens $g - j$ and $h - k$ are uniformly distributed in $[\pm n]$. In particular, the pair $(g - j, h - k)$ is uniformly distributed in $[\pm n]^2$. Now recall that $i \in P$, and that as a consequence $i$ is involved in at least $\Omega(n^2 / L)$ additive quadruples $(i = i_1, i_2, i_3, i_4)$ by \cref{lem:conv-decomp-quasi-pop}~(i). Thus, with probability at least $\Omega(1/L)$, we hit one such tuple with $i_3 = g - j$ and $i_4 = h - k$. In this case, we also have $g + h = i_3 + i_4 + j + k = i_1 + i_2 + j + k = i_2$, recalling that $i + j + k = 0$, and thus \cref{eq:fredman-energy} holds.

We now argue that in any successful iteration, the solution $(i, j, k)$ is listed. Indeed, combining~\eqref{eq:fredman-energy} with the fact that $(i, j, k)$ is a solution, $a[i] + a[j] + a[k] = 0$, we get that
\begin{align*}
    b[j] + c[k]
    &= a[j] + a[g - j] + a[k] + a[h - k] \\
    &= a[g - j] + a[h - k] - a[i] \\
    &= a[g + h] \pm 2\Delta.
\end{align*}
Thus, the pair $(j, k)$ satisfies Equation~\eqref{eq:fredman} and we report $(i, j, k)$.

\paragraph{Running Time}
Finally, focus on the running time. The algorithm runs in $\Theta(L \log n)$ iterations, and the running time of each iteration is dominated by listing the solutions to Equation~\eqref{eq:fredman}, in time proportional to the number of solutions. We argue that the expected number of solution in each iteration is at most $O(n^2 / M^{1/2})$. From this the claimed time bound follows (at least in expectation; to achieve an upper bound with high probability one may repeat the algorithm a logarithmic number of times and interrupt the excessively expensive runs). Indeed, note that each solution to Equation~\eqref{eq:fredman} is of the form
\begin{equation*}
    a[j] + a[g - j] + a[k] + a[h - k] = a[g + h] \pm 2\Delta,
\end{equation*}
or equivalently, by renaming $p = g - j$ and $q = h - k$,
\begin{equation} \label{eq:big-energy}
    \begin{cases}
        j + k + p + q = r, \\
        a[j] + a[k] + a[p] + a[q] = a[r] \pm 2\Delta.
    \end{cases}
\end{equation}
Using the fact that we only enumerate solutions where $j$ is unpopular and applying \Cref{lem:few-solutions-unpop}, the total number of solutions to Equation~\eqref{eq:big-energy} is bounded by $n^4 / M^{1/2}$. This leaves $n^2 / M^{1/2}$ solutions in expectation over the choice of $g$ and $h$.
\end{proof}

\subsection{Type 4: Iterate}\label{sec:type4}

Solutions of Type 4 are not found in the current iteration. Instead, notice that any solution containing at least one element of $P$ is covered by Types 1, 2, and 3. Hence, we can remove all elements in $P$ and iterate, starting from finding a new threshold $\Delta$.
At some point, even the number of $\infty$-approximate solutions becomes truly subquadratic. Then, we can simply list all triples where no element has been removed, using the listing-to-decision reduction.

\begin{proof}[Proof of \Cref{thm:conv-to-integer-reduction}]
The algorithm works in iterations. In one iteration, we first compute the approximation threshold $\Delta$ with \Cref{lem:approx-threshold}.
Then, we compute the partition $[\pm n] = P \sqcup Q \sqcup U$ using \Cref{lem:conv-decomp-quasi-pop}.
Next, we check for solutions involving at least one index from $P$ by invoking \Cref{lem:type1,lem:type2,lem:type3}.
Afterwards, we remove $P$ from $a$, i.e., replace all entries $a[i]$ with $i \in P$ by $\bot$.
Finally, when \Cref{lem:approx-threshold} reports that there are at most $O(n^2 / K)$ $\infty$-approximate solutions, we list all solutions $i + j + k = 0$ for indices $i,j,k \in [\pm n]$ corresponding to non-$\bot$ entries using \Cref{lem:3sum-listing-to-decision} in time $O(n^2/ K^{\epsilon/2})$.

Let $I$ denote the total number of iterations and $P_i$ the set of popular items in the $i$-th iteration. Then the total time over all iterations is
\begin{equation*}
    O\parens*{\frac{n^2}{K^{\epsilon/2}}} + \sum_{i=1}^I \tilde O\parens*{n (K L M)^{O(1)} + \frac{|P_i|n}{K^{\epsilon / 2}} + \frac{n^2}{L^{\epsilon / 4}} + \frac{n^2L}{M^{1/2}}}.
\end{equation*}

It remains to bound $I$. To this end, we analyze the number of elements removed in each iteration,~$\abs{P}$. Here we apply \cref{lem:approx-sol-count} stating that the total number of $2\Delta$-approximate solutions of type 1 is at least $\Omega(n^2 / K)$. Recall that these are the solutions involving \emph{only} popular items. But the number of such solutions is clearly at most $|P|^2$, entailing that $|P| \geq \Omega(n / \sqrt{K})$. Thus, the number of iterations is at most $I \leq O(\sqrt K)$. With this bound, the previous bound on the running time turns into 
\begin{equation*}
    \tilde O\parens*{n(KLM)^{O(1)} + \frac{n^2}{K^{\epsilon / 2}} + \frac{n^2K^{1/2}}{L^{\epsilon/4}} + \frac{n^2LK^{1/2}}{M^{1/2}}}.
\end{equation*}
Choosing $L = K^{4/\epsilon}$ (which certainly satisfies the premise of \cref{lem:approx-sol-count}) and $M = L^2K^2$, this becomes
\begin{equation*}
    \tilde O\parens*{n K^{O(1/\epsilon)} + \frac{n^2}{K^{\Omega(\epsilon)}}}.
\end{equation*}
This is truly subquadratic, $n^{2-\Omega(\epsilon^2)}$, for some small enough choice $K = n^{\Theta(\epsilon)}$.
\end{proof}

Finally, \cref{thm:main} follows by concatenating the reductions from \cref{thm:real-to-real-conv,thm:conv-to-integer-reduction}. Having made sure that our reduction can be implemented in the realistic comparison-addition model (which mostly concerns the Frank--Tardos theorem in \cref{sec:toolkit:sec:frank-tardos:sec:realistic}), \cref{thm:equiv} also immediately follows.

%% file: main.bib
@inproceedings{Patrascu10,
  author       = {Mihai P{\u{a}}tra{\c{s}}cu},
  _editor       = {Leonard J. Schulman},
  title        = {Towards polynomial lower bounds for dynamic problems},
  booktitle    = {42nd Annual {ACM} Symposium on Theory of Computing ({STOC} 2010)},
  pages        = {603--610},
  publisher    = {{ACM}},
  year         = {2010},
  url          = {https://doi.org/10.1145/1806689.1806772},
  doi          = {10.1145/1806689.1806772},
  timestamp    = {Wed, 14 Nov 2018 10:51:36 +0100},
  biburl       = {https://dblp.org/rec/conf/stoc/Patrascu10.bib},
  bibsource    = {dblp computer science bibliography, https://dblp.org}
}

@inproceedings{ChanL15,
  author       = {Timothy M. Chan and
                  Moshe Lewenstein},
  _editor       = {Rocco A. Servedio and
                  Ronitt Rubinfeld},
  title        = {Clustered Integer {3SUM} via Additive Combinatorics},
  booktitle    = {47th Annual {ACM} Symposium on Theory of Computing ({STOC} 2015)},
  pages        = {31--40},
  publisher    = {{ACM}},
  year         = {2015},
  url          = {https://doi.org/10.1145/2746539.2746568},
  doi          = {10.1145/2746539.2746568},
  timestamp    = {Tue, 06 Nov 2018 11:07:04 +0100},
  biburl       = {https://dblp.org/rec/conf/stoc/ChanL15.bib},
  bibsource    = {dblp computer science bibliography, https://dblp.org}
}

@inproceedings{DudekGS20,
  author       = {Bartłomiej Dudek and
                  Paweł Gawrychowski and
                  Tatiana Starikovskaya},
  _editor       = {Konstantin Makarychev and
                  Yury Makarychev and
                  Madhur Tulsiani and
                  Gautam Kamath and
                  Julia Chuzhoy},
  title        = {All non-trivial variants of 3-LDT are equivalent},
  booktitle    = {52nd Annual {ACM} Symposium on Theory of Computing ({STOC} 2020)},
  pages        = {974--981},
  publisher    = {{ACM}},
  year         = {2020},
  url          = {https://doi.org/10.1145/3357713.3384275},
  doi          = {10.1145/3357713.3384275},
  timestamp    = {Mon, 18 Dec 2023 07:33:37 +0100},
  biburl       = {https://dblp.org/rec/conf/stoc/0001GS20.bib},
  bibsource    = {dblp computer science bibliography, https://dblp.org}
}

@inproceedings{ChanVX22,
  author       = {Timothy M. Chan and
                  Vassilevska Williams, Virginia and
                  Yinzhan Xu},
  _editor       = {Stefano Leonardi and
                  Anupam Gupta},
  title        = {Hardness for triangle problems under even more believable hypotheses:
                  reductions from real {APSP}, real {3SUM}, and {OV}},
  booktitle    = {54th Annual {ACM} Symposium on Theory of Computing ({STOC} 2022)},
  pages        = {1501--1514},
  publisher    = {{ACM}},
  year         = {2022},
  url          = {https://doi.org/10.1145/3519935.3520032},
  doi          = {10.1145/3519935.3520032},
  timestamp    = {Tue, 27 Dec 2022 09:06:31 +0100},
  biburl       = {https://dblp.org/rec/conf/stoc/ChanWX22.bib},
  bibsource    = {dblp computer science bibliography, https://dblp.org}
}

@inproceedings{AbboudBF23,
  author       = {Amir Abboud and
                  Karl Bringmann and
                  Nick Fischer},
  _editor       = {Barna Saha and
                  Rocco A. Servedio},
  title        = {Stronger 3-{SUM} Lower Bounds for Approximate Distance Oracles via Additive
                  Combinatorics},
  booktitle    = {55th Annual {ACM} Symposium on Theory of Computing ({STOC} 2023)},
  pages        = {391--404},
  publisher    = {{ACM}},
  year         = {2023},
  url          = {https://doi.org/10.1145/3564246.3585240},
  doi          = {10.1145/3564246.3585240},
  timestamp    = {Mon, 22 May 2023 13:15:17 +0200},
  biburl       = {https://dblp.org/rec/conf/stoc/AbboudBF23.bib},
  bibsource    = {dblp computer science bibliography, https://dblp.org}
}

@inproceedings{ChanVX23,
  author       = {Timothy M. Chan and
                  Vassilevska Williams, Virginia and
                  Yinzhan Xu},
  _editor       = {Barna Saha and
                  Rocco A. Servedio},
  title        = {{Fredman}'s Trick Meets Dominance Product: Fine-Grained Complexity of
                  Unweighted {APSP}, {3SUM} Counting, and More},
  booktitle    = {55th Annual {ACM} Symposium on Theory of Computing ({STOC} 2023)},
  pages        = {419--432},
  publisher    = {{ACM}},
  year         = {2023},
  url          = {https://doi.org/10.1145/3564246.3585237},
  doi          = {10.1145/3564246.3585237},
  timestamp    = {Sun, 19 Jan 2025 13:28:33 +0100},
  biburl       = {https://dblp.org/rec/conf/stoc/ChanWX23.bib},
  bibsource    = {dblp computer science bibliography, https://dblp.org}
}

@inproceedings{JinX23,
  author       = {Ce Jin and
                  Yinzhan Xu},
  _editor       = {Barna Saha and
                  Rocco A. Servedio},
  title        = {Removing Additive Structure in {3SUM}-Based Reductions},
  booktitle    = {55th Annual {ACM} Symposium on Theory of Computing ({STOC} 2023)},
  pages        = {405--418},
  publisher    = {{ACM}},
  year         = {2023},
  url          = {https://doi.org/10.1145/3564246.3585157},
  doi          = {10.1145/3564246.3585157},
  timestamp    = {Sun, 19 Jan 2025 13:28:39 +0100},
  biburl       = {https://dblp.org/rec/conf/stoc/JinX23.bib},
  bibsource    = {dblp computer science bibliography, https://dblp.org}
}

@inproceedings{VassilevskaX20b,
  author       = {Vassilevska Williams, Virginia and
                  Yinzhan Xu},
  _editor       = {Sandy Irani},
  title        = {Monochromatic Triangles, Triangle Listing and {APSP}},
  booktitle    = {61st Annual {IEEE} Symposium on Foundations of Computer Science ({FOCS} 2020)},
  pages        = {786--797},
  publisher    = {{IEEE}},
  year         = {2020},
  url          = {https://doi.org/10.1109/FOCS46700.2020.00078},
  doi          = {10.1109/FOCS46700.2020.00078},
  timestamp    = {Wed, 15 Dec 2021 10:49:16 +0100},
  biburl       = {https://dblp.org/rec/conf/focs/WilliamsX20.bib},
  bibsource    = {dblp computer science bibliography, https://dblp.org}
}

@inproceedings{KopelowitzPP16,
  author       = {Tsvi Kopelowitz and
                  Seth Pettie and
                  Ely Porat},
  _editor       = {Robert Krauthgamer},
  title        = {Higher Lower Bounds from the {3SUM} Conjecture},
  booktitle    = {27th Annual {ACM-SIAM} Symposium on Discrete Algorithms ({SODA} 2016)},
  pages        = {1272--1287},
  publisher    = {{SIAM}},
  year         = {2016},
  _url          = {https://doi.org/10.1137/1.9781611974331.ch89},
  doi          = {10.1137/1.9781611974331.CH89},
  timestamp    = {Tue, 02 Feb 2021 17:07:56 +0100},
  biburl       = {https://dblp.org/rec/conf/soda/KopelowitzPP16.bib},
  bibsource    = {dblp computer science bibliography, https://dblp.org}
}

@inproceedings{Fischer24,
  author       = {Nick Fischer},
  _editor       = {David P. Woodruff},
  title        = {Deterministic Sparse Pattern Matching via the {Baur}-{Strassen} Theorem},
  booktitle    = {35th Annual {ACM-SIAM} Symposium on Discrete Algorithms ({SODA} 2024)},
  pages        = {3333--3353},
  publisher    = {{SIAM}},
  year         = {2024},
  url          = {https://doi.org/10.1137/1.9781611977912.119},
  doi          = {10.1137/1.9781611977912.119},
  timestamp    = {Thu, 04 Apr 2024 19:38:10 +0200},
  biburl       = {https://dblp.org/rec/conf/soda/Fischer24.bib},
  bibsource    = {dblp computer science bibliography, https://dblp.org}
}

@inproceedings{Fischer25,
  author       = {Nick Fischer},
  _editor       = {Yossi Azar and
                  Debmalya Panigrahi},
  title        = {Sumsets, {3SUM}, Subset Sum: Now for Real!},
  booktitle    = {36th Annual {ACM-SIAM} Symposium on Discrete Algorithms ({SODA} 2025)},
  pages        = {4520--4546},
  publisher    = {{SIAM}},
  year         = {2025},
  url          = {https://doi.org/10.1137/1.9781611978322.155},
  doi          = {10.1137/1.9781611978322.155},
  timestamp    = {Tue, 28 Jan 2025 14:38:41 +0100},
  biburl       = {https://dblp.org/rec/conf/soda/Fischer25.bib},
  bibsource    = {dblp computer science bibliography, https://dblp.org}
}

@inproceedings{FischerJX25,
  author       = {Nick Fischer and
                  Ce Jin and
                  Yinzhan Xu},
  _editor       = {Yossi Azar and
                  Debmalya Panigrahi},
  title        = {New Applications of {3SUM}-Counting in Fine-Grained Complexity and Pattern
                  Matching},
  booktitle    = {36th Annual {ACM-SIAM} Symposium on Discrete Algorithms ({SODA} 2025)},
  pages        = {4547--4595},
  publisher    = {{SIAM}},
  year         = {2025},
  url          = {https://doi.org/10.1137/1.9781611978322.156},
  doi          = {10.1137/1.9781611978322.156},
  timestamp    = {Tue, 28 Jan 2025 14:38:41 +0100},
  biburl       = {https://dblp.org/rec/conf/soda/Fischer0X25.bib},
  bibsource    = {dblp computer science bibliography, https://dblp.org}
}

@inproceedings{DurajKPW20,
  author       = {Lech Duraj and Krzysztof Kleiner and Adam Polak and Vassilevska Williams, Virginia},
  title        = {Equivalences between triangle and range query problems},
  booktitle    = {31st Annual {ACM-SIAM} Symposium on Discrete Algorithms ({SODA} 2020)},
  pages        = {30--47},
  publisher    = {{SIAM}},
  year         = {2020},
  doi          = {10.1137/1.9781611975994.3},
}

@inproceedings{ChanX24,
  author       = {Timothy M. Chan and
                  Yinzhan Xu},
  _editor       = {Merav Parter and
                  Seth Pettie},
  title        = {Simpler Reductions from Exact Triangle},
  booktitle    = {7th Symposium on Simplicity in Algorithms ({SOSA} 2024)},
  pages        = {28--38},
  publisher    = {{SIAM}},
  year         = {2024},
  url          = {https://doi.org/10.1137/1.9781611977936.4},
  doi          = {10.1137/1.9781611977936.4},
  timestamp    = {Wed, 10 Apr 2024 20:26:09 +0200},
  biburl       = {https://dblp.org/rec/conf/sosa/ChanX24.bib},
  bibsource    = {dblp computer science bibliography, https://dblp.org}
}

@inproceedings{ChanH20,
  author       = {Timothy M. Chan and
                  Qizheng He},
  title        = {Reducing {3SUM} to Convolution-{3SUM}},
  booktitle    = {3rd Symposium on Simplicity in Algorithms, {SOSA} 2020},
  pages        = {1--7},
  publisher    = {{SIAM}},
  year         = {2020},
  doi          = {10.1137/1.9781611976014.1},
}

@inproceedings{LincolnWWW16,
  author       = {Andrea Lincoln and Vassilevska Williams, Virginia and Joshua R. Wang and R. Ryan Williams},
  title        = {Deterministic Time-Space Trade-Offs for k-{SUM}},
  booktitle    = {43rd International Colloquium on Automata, Languages, and Programming, {ICALP} 2016},
  series       = {LIPIcs},
  pages        = {58:1--58:14},
  publisher    = {Schloss Dagstuhl - Leibniz-Zentrum f{\"{u}}r Informatik},
  year         = {2016},
  doi          = {10.4230/LIPICS.ICALP.2016.58},
}

@inproceedings{RandolphW24,
  author       = {Tim Randolph and
                  Karol Węgrzycki},
  _editor       = {Timothy M. Chan and
                  Johannes Fischer and
                  John Iacono and
                  Grzegorz Herman},
  title        = {Parameterized Algorithms on Integer Sets with Small Doubling: Integer
                  Programming, Subset Sum and {$k$}-{SUM}},
  booktitle    = {32nd Annual European Symposium on Algorithms ({ESA} 2024)},
  series       = {LIPIcs},
  volume       = {308},
  pages        = {96:1--96:19},
  publisher    = {Schloss Dagstuhl - Leibniz-Zentrum f{\"{u}}r Informatik},
  year         = {2024},
  url          = {https://doi.org/10.4230/LIPICS.ESA.2024.96},
  doi          = {10.4230/LIPICS.ESA.2024.96},
  timestamp    = {Mon, 14 Oct 2024 16:58:03 +0200},
  biburl       = {https://dblp.org/rec/conf/esa/0001W24.bib},
  bibsource    = {dblp computer science bibliography, https://dblp.org}
}

@inproceedings{GoldS17,
  author       = {Omer Gold and Micha Sharir},
  title        = {Improved Bounds for {3SUM}, k-{SUM}, and Linear Degeneracy},
  booktitle    = {25th Annual European Symposium on Algorithms, {ESA} 2017},
  series       = {LIPIcs},
  pages        = {42:1--42:13},
  publisher    = {Schloss Dagstuhl - Leibniz-Zentrum f{\"{u}}r Informatik},
  year         = {2017},
  doi          = {10.4230/LIPICS.ESA.2017.42},
}

@inproceedings{FischerKP24,
  author       = {Nick Fischer and
                  Piotr Kaliciak and
                  Adam Polak},
  title        = {Deterministic {3SUM}-Hardness},
  booktitle    = {15th Innovations in Theoretical Computer Science Conference, {ITCS} 2024},
  series       = {LIPIcs},
  pages        = {49:1--49:24},
  publisher    = {Schloss Dagstuhl - Leibniz-Zentrum f{\"{u}}r Informatik},
  year         = {2024},
  doi          = {10.4230/LIPICS.ITCS.2024.49},
}

@inbook{VW18,
  author    = {Vasillevska Williams, Virginia},
  title     = {On some fine-grained questions in algorithms and complexity},
  booktitle = {Proceedings of the International Congress of Mathematicians (ICM 2018)},
  pages     = {3447-3487},
  doi       = {10.1142/9789813272880_0188},
  year      = {2018},
  publisher = {World Scientific},
}

@inproceedings{Dietzfelbinger96,
  author       = {Martin Dietzfelbinger},
  _editor       = {Claude Puech and
                  R{\"{u}}diger Reischuk},
  title        = {Universal Hashing and k-Wise Independent Random Variables via Integer
                  Arithmetic without Primes},
  booktitle    = {13th Annual Symposium on Theoretical Aspects of Computer Science ({STACS} 1996)},
  series       = {Lecture Notes in Computer Science},
  pages        = {569--580},
  publisher    = {Springer},
  year         = {1996},
  url          = {https://doi.org/10.1007/3-540-60922-9\_46},
  doi          = {10.1007/3-540-60922-9\_46},
  timestamp    = {Tue, 14 May 2019 10:00:48 +0200},
  biburl       = {https://dblp.org/rec/conf/stacs/Dietzfelbinger96.bib},
  bibsource    = {dblp computer science bibliography, https://dblp.org}
}

@inproceedings{HartmanisS74,
  author       = {Juris Hartmanis and Janos Simon},
  title        = {On the Power of Multiplication in Random Access Machines},
  booktitle    = {15th Annual Symposium on Switching and Automata Theory ({SWAT} 1974)},
  pages        = {13--23},
  publisher    = {{IEEE} Computer Society},
  year         = {1974},
  doi          = {10.1109/SWAT.1974.20},
}

@article{VassilevskaW18,
  author       = {Vassilevska Williams, Virginia and
                  R. Ryan Williams},
  title        = {Subcubic Equivalences Between Path, Matrix, and Triangle Problems},
  journal      = {J. {ACM}},
  volume       = {65},
  number       = {5},
  pages        = {27:1--27:38},
  year         = {2018},
  url          = {https://doi.org/10.1145/3186893},
  doi          = {10.1145/3186893},
  bibsource    = {dblp computer science bibliography, https://dblp.org}
}

@article{GronlundP18,
  note         = {Announced at {FOCS} 2014},
  author       = {Allan Gr{\o}nlund and Seth Pettie},
  title        = {Threesomes, Degenerates, and Love Triangles},
  journal      = {J. {ACM}},
  volume       = {65},
  number       = {4},
  pages        = {22:1--22:25},
  year         = {2018},
  doi          = {10.1145/3185378},
}

@article{Fredman76,
  author       = {Michael L. Fredman},
  title        = {New Bounds on the Complexity of the Shortest Path Problem},
  journal      = {{SIAM} J. Comput.},
  volume       = {5},
  number       = {1},
  pages        = {83--89},
  year         = {1976},
  url          = {https://doi.org/10.1137/0205006},
  doi          = {10.1137/0205006},
  bibsource    = {dblp computer science bibliography, https://dblp.org}
}

@article{WilliamsW13,
  note         = {Announced at {STOC} 2009},
  author       = {Vassilevska Williams, Virginia and Ryan Williams},
  title        = {Finding, Minimizing, and Counting Weighted Subgraphs},
  journal      = {{SIAM} J. Comput.},
  volume       = {42},
  number       = {3},
  pages        = {831--854},
  year         = {2013},
  doi          = {10.1137/09076619X},
}

@article{Chan20,
  note         = {Announced at {SODA} 2018},
  author       = {Timothy M. Chan},
  title        = {More Logarithmic-factor Speedups for {3SUM}, (median,+)-convolution, and Some Geometric {3SUM}-hard Problems},
  journal      = {{ACM} Trans. Algorithms},
  volume       = {16},
  number       = {1},
  pages        = {7:1--7:23},
  year         = {2020},
  doi          = {10.1145/3363541},
}

@article{Freund17,
  author       = {Ari Freund},
  title        = {Improved Subquadratic {3SUM}},
  journal      = {Algorithmica},
  volume       = {77},
  number       = {2},
  pages        = {440--458},
  year         = {2017},
  doi          = {10.1007/S00453-015-0079-6},
}

@article{BaranDP08,
  note         = {Announced at {WADS} 2005},
  author       = {Ilya Baran and
                  Erik D. Demaine and
                  Mihai P{\u{a}}tra{\c{s}}cu},
  title        = {Subquadratic Algorithms for {3SUM}},
  journal      = {Algorithmica},
  volume       = {50},
  number       = {4},
  pages        = {584--596},
  year         = {2008},
  doi          = {10.1007/S00453-007-9036-3},
}

@article{Yuval76,
  author       = {Gideon Yuval},
  title        = {An Algorithm for Finding All Shortest Paths Using N{\^{}}(2.81) Infinite-Precision
                  Multiplications},
  journal      = {Inf. Process. Lett.},
  volume       = {4},
  number       = {6},
  pages        = {155--156},
  year         = {1976},
  doi          = {10.1016/0020-0190(76)90085-5},
}

@article{Boruvka28,
  author       = {Otakar Borůvka},
  title        = {Příspěvek k řešení otázky ekonomické stavby elektrovodních sítí},
  journal      = {Elektronický Obzor},
  year         = {1928},
  volume       = {15},
  pages        = {153--154}
}

@misc{Freiman66,
  author       = {Gregory A. Fre{\u{\i}}man},
  title        = {\foreignlanguage{russian}{Начала структурной теории сложения множеств}},
  year         = {1966},
  language     = {Russian},
  howpublished = {Kazan': {Kazan}. {Gosudarstv}. {Ped}. {Inst}.},
}

@article{FrankT87,
  author       = {Andr{\'{a}}s Frank and
                  {\'{E}}va Tardos},
  title        = {An application of simultaneous Diophantine approximation in combinatorial
                  optimization},
  journal      = {Comb.},
  volume       = {7},
  number       = {1},
  pages        = {49--65},
  year         = {1987},
  url          = {https://doi.org/10.1007/BF02579200},
  doi          = {10.1007/BF02579200},
  bibsource    = {dblp computer science bibliography, https://dblp.org}
}

@article{Ruzsa94,
  author       = {Imre Z. Ruzsa},
  title        = {Generalized arithmetical progressions and sumsets},
  journal      = {Acta Mathematica Hungarica},
  volume       = {65},
  number       = {4},
  pages        = {379--388},
  year         = {1994},
  url          = {https://doi.org/10.1007/BF01876039},
  doi          = {10.1007/BF01876039},
}

@article{GajentaanO95,
  author       = {Anka Gajentaan and
                  Mark H. Overmars},
  title        = {On a Class of {$O(n^2)$} Problems in Computational Geometry},
  journal      = {Comput. Geom.},
  volume       = {5},
  pages        = {165--185},
  year         = {1995},
  url          = {https://doi.org/10.1016/0925-7721(95)00022-2},
  doi          = {10.1016/0925-7721(95)00022-2},
  bibsource    = {dblp computer science bibliography, https://dblp.org}
}

@article{Bilu99,
  author       = {Yuri Bilu},
  title        = {Structure of sets with small sumset},
  journal      = {Astérisque},
  year         = {1999},
  volume       = {258},
  pages        = {77--108},
}

@article{Chang02,
  author       = {Mei-Chu Chang},
  title        = {A polynomial bound in {F}reiman’s theorem},
  journal      = {Duke Mathematical Journal},
  year         = {2002},
  volume       = {113},
  number       = {3},
  url          = {https://doi.org/10.1215/S0012-7094-02-11331-3},
  doi          = {10.1215/S0012-7094-02-11331-3},
}

@article{GreenR07,
  author       = {Ben Green and
                  Imre Z. Ruzsa},
  title        = {Freiman’s theorem in an arbitrary abelian group},
  journal      = {Journal of the London Mathematical Society},
  volume       = {75},
  number       = {1},
  pages        = {163--175},
  year         = {2007},
  url          = {https://doi.org/10.1112/jlms/jdl021},
  doi          = {10.1112/jlms/jdl021},
}

@article{Sanders12,
  author       = {Tom Sanders},
  title        = {On the {Bogolyubov}–{Ruzsa} lemma},
  journal      = {Analysis \& PDE},
  volume       = {5},
  number       = {3},
  pages        = {627--655},
  year         = {2012},
  publisher    = {Mathematical Sciences Publishers},
  url          = {https://doi.org/10.2140/apde.2012.5.627},
  doi          = {10.2140/apde.2012.5.627},
}

@article{GowersGMT25,
  author       = {W.T. Gowers and
                  Ben Green and
                  Freddie Manners and
                  Terence Tao},
  title        = {On a conjecture of {Marton}},
  journal      = {Annals of Mathematics},
  volume       = {201},
  number       = {2},
  pages        = {515--549},
  year         = {2025},
  url          = {https://doi.org/10.4007/annals.2025.201.2.5},
  doi          = {10.4007/annals.2025.201.2.5},
}

@article{LenstraLL1982,
  author  = {Lenstra, A. K. and
             Lenstra, H. W. and
             Lov{\'a}sz, L.},
  title   = {Factoring polynomials with rational coefficients},
  journal = {Mathematische Annalen},
  year    = {1982},
  volume  = {261},
  number  = {4},
  pages   = {515--534},
  doi     = {10.1007/BF01457454},
}

@book{TaoV06,
  author       = {Terence Tao and
                  Van H. Vu},
  title        = {Additive Combinatorics},
  publisher    = {Cambridge University Press},
  year         = {2006},
  url          = {https://doi.org/10.1017/cbo9780511755149},
  doi          = {10.1017/cbo9780511755149},
}

@incollection{Eisenbrand10,
  author       = {Friedrich Eisenbrand},
  _editor       = {Michael J{\"{u}}nger and
                  Thomas M. Liebling and
                  Denis Naddef and
                  George L. Nemhauser and
                  William R. Pulleyblank and
                  Gerhard Reinelt and
                  Giovanni Rinaldi and
                  Laurence A. Wolsey},
  title        = {Integer Programming and Algorithmic Geometry of Numbers - {A} tutorial},
  booktitle    = {50 Years of Integer Programming 1958-2008 - {From} the Early Years to
                  the State-of-the-Art},
  pages        = {505--559},
  publisher    = {Springer},
  year         = {2010},
  url          = {https://doi.org/10.1007/978-3-540-68279-0\_14},
  doi          = {10.1007/978-3-540-68279-0\_14},
  timestamp    = {Wed, 14 Nov 2018 10:12:21 +0100},
  biburl       = {https://dblp.org/rec/books/daglib/p/Eisenbrand10.bib},
  bibsource    = {dblp computer science bibliography, https://dblp.org}
}

@book{Vazirani01,
  title        = {Approximation algorithms},
  author       = {Vazirani, Vijay V},
  year         = {2001},
  publisher    = {Springer},
  url          = {https://doi.org/10.1007/978-3-662-04565-7},
  doi          = {10.1007/978-3-662-04565-7},
}
